\documentclass[12pt]{amsart}
\usepackage{amsmath,amsfonts,amssymb,amsthm,epsfig}
\usepackage{latexsym}
\usepackage{amssymb}
\usepackage{enumerate}
\usepackage{graphicx}
\usepackage{pstricks}
\input xy
\xyoption{all}
\usepackage{comment}

\usepackage{hyperref}

\def\ii{{\mathfrak i}}

\def\ww{{\mathfrak w}}
\def\xx{{\mathfrak x}}
\def\yy{{\mathfrak y}}

\def\DD{{\mathfrak D}}

\def\UU{{\mathfrak U}}

\def\CCC{{\mathbb C}}

\def\RRR{{\mathbb R}}

\def\ZZZ{{\mathbb Z}}

\def\ga{{\alpha}}
\def\gb{{\beta}}

\def\gd{{\delta}}

\def\gl{{\lambda}}
\def\gm{{\mu}}
\def\gn{{\nu}}

\def\gq{{\chi}}

\def\gt{{\tau}}

\def\gD{{\Delta}}

\def\({\left(}
\def\){\right)}
\def\[{\left[}
\def\]{\right]}
\def\<{\left\langle}
\def\>{\right\rangle}

\def\union{\cup}

\def\qed{\hfill\Box\smallskip}

\newtheorem{theorem}{Theorem}[section]
\newtheorem{lemma}[theorem]{Lemma}

\newtheorem{corollary}[theorem]{Corollary}
\newtheorem{proposition}[theorem]{Proposition}

\begin{document}

\title[Arbitrary slopes]
{Scaling limits of random skew plane partitions with arbitrarily sloped back walls}

\author[Sevak Mkrtchyan]{Sevak Mkrtchyan}
\email{sevak.mkrtchyan@rice.edu}
\address{Rice University, Department of Mathematics\\ Houston, TX}

\begin{abstract}

The paper studies scaling limits of random skew plane partitions confined to a box when the inner shapes converge uniformly to a piecewise linear function $V$ of arbitrary slopes in $[-1,1]$. It is shown that the correlation kernels in the bulk are given by the incomplete Beta kernel, as expected. As a consequence it is established that the local correlation functions in the scaling limit do not depend on the particular sequence of discrete inner shapes that converge to $V$. A detailed analysis of the correlation kernels at the top of the limit shape, and of the frozen boundary is given. It is shown that depending on the slope of the linear section of the back wall, the system exhibits behavior observed in either \cite{OR2} or \cite{BMRT}. 

\end{abstract}

\maketitle

\tableofcontents

\section{Introduction}
\subsection{Background}
\subsubsection{Skew plane partitions}
Given a partition $\lambda$, a skew plane partition with boundary $\lambda$ confined to a $c\times d$ box is an array of nonnegative integers $\pi=\{\pi_{i,j}\}$ defined for all $1\leq i\leq c$, $1\leq j\leq d$, $(i,j)\notin\lambda$, which are non-increasing in $i$ and $j$. One way to visualize this is to draw a $c\times d$ rectangular grid, remove the partition $\lambda$ from a corner, and stack $\pi_{i,j}$ identical cubes at position $(i,j)$, as in Figure \ref{Fig:3d_skew_part}. The number of cubes $|\pi|:=\sum\pi_{i,j}$ is the volume of the skew plane partition $\pi$.
\begin{figure}[ht]
\caption{\label{Fig:3d_skew_part} A skew plane partition and the corresponding partition $\gl$. Here $\gl=\{4,4,3,3,3,1\}$.}
\includegraphics[width=10cm]{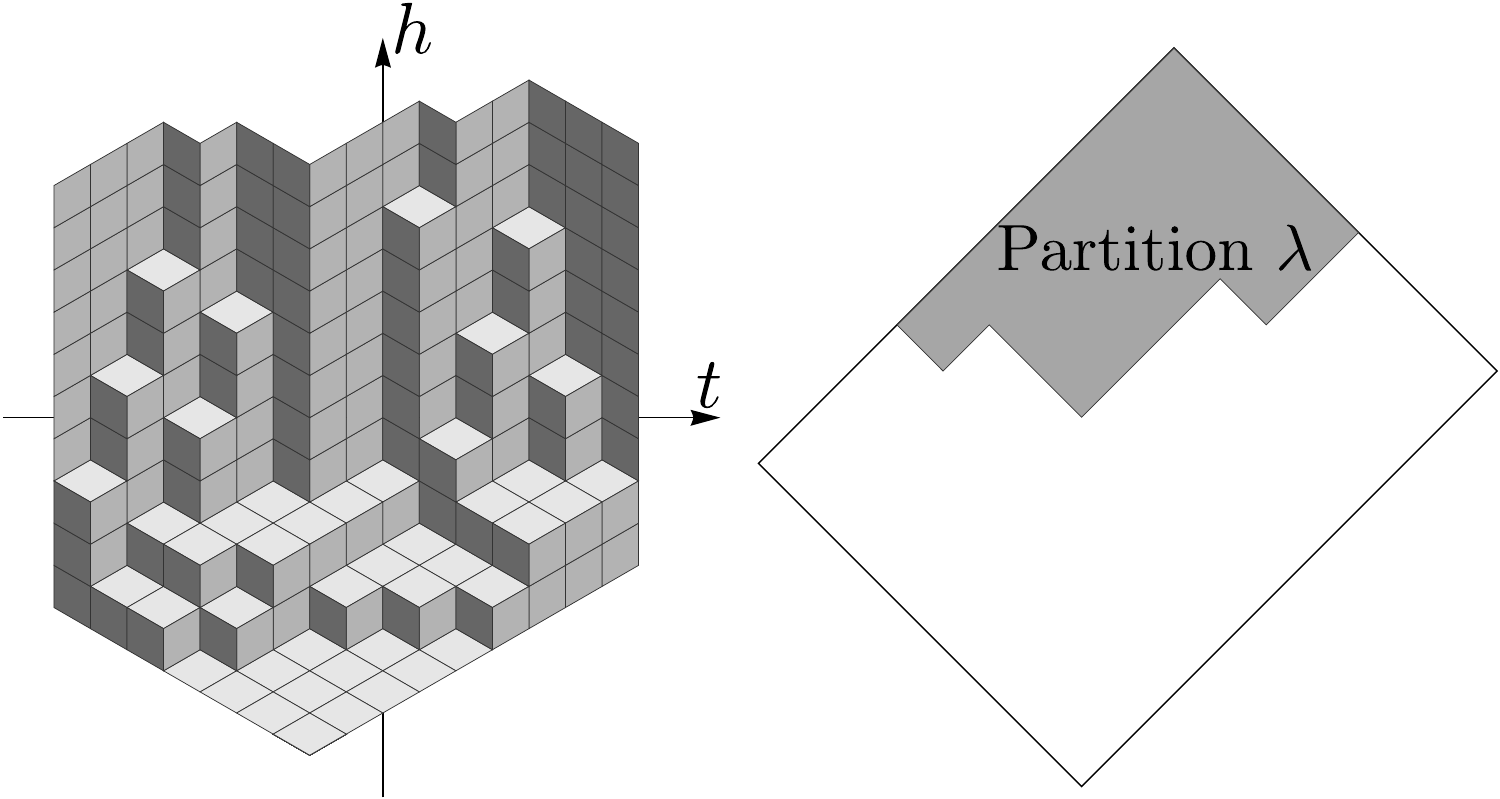}
\end{figure}

From Figure \ref{Fig:3d_skew_part} it is easy to see that skew plane partitions can be identified with tilings of a certain region of $\mathbb{R}^2$ with 3 types of rhombi (see \cite{OR1} or \cite{OR2} for details and other correspondences). Scale the axes in such a way that the centers of horizontal tiles are on the lattice $\mathbb{Z}\times \frac 12 \mathbb{Z}$. The letter $t$ will be used for the horizontal and $h$ for the vertical coordinate axes in this plane. For a partition $\gl$, let $b_\gl(t)$ encode the boundary of $\lambda$. More precisely, let $u_1<u_2<\ldots<u_{n-1}$ denote the $t$ coordinates of the corners on the outer boundary of the Young diagram $\lambda$ as shown in Figure \ref{Fig:corners}. Define $b_\lambda(t)$ as 
\begin{figure}[ht]
\caption{\label{Fig:corners} Position of corners for $\gl=\{4,4,3,3,3,1\}$.}
\includegraphics[width=7cm]{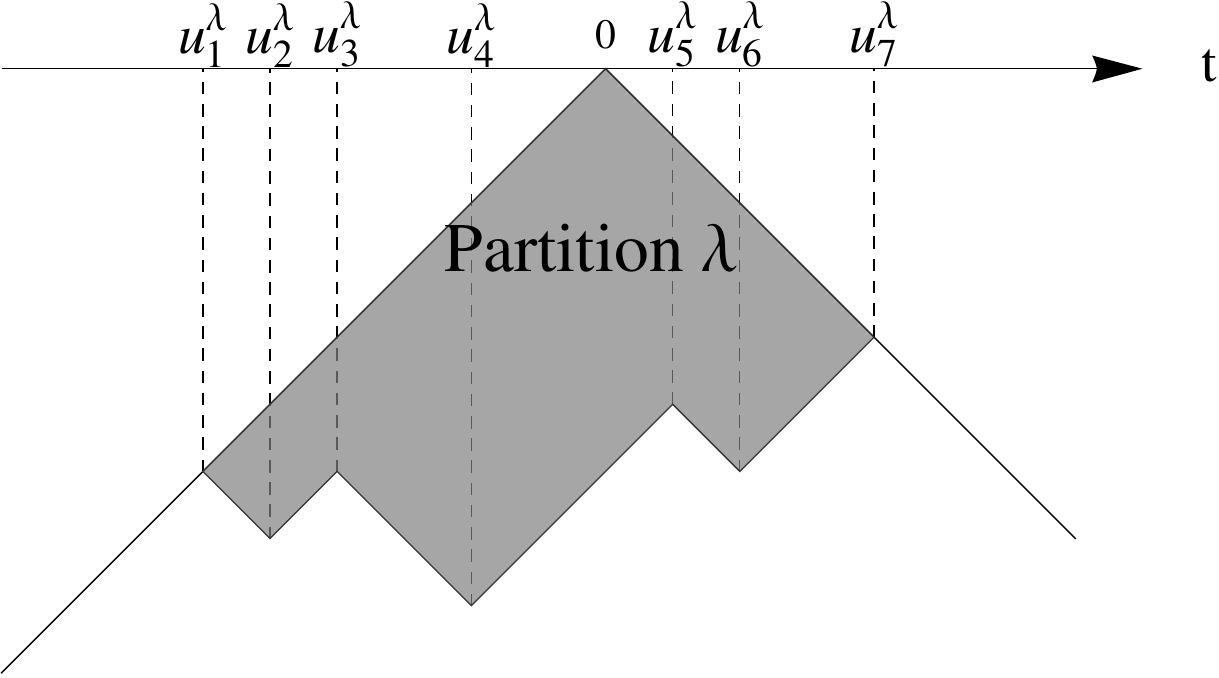}
\end{figure}
\begin{equation*}
b_\lambda(t):=t+2\sum_{i=1}^{n-1} (-1)^{i}(t-u_i)\theta(t-u_i),
\end{equation*}
where $\theta(t)$ is the step function 
\begin{equation*}
\theta(t)=\left\{\begin{array}{rr}
1,&t\geq 0\\0,&t<0
\end{array}\right..
\end{equation*}
We will call $b_\lambda(t)$ the back wall. Notice that $b_\lambda(t)$ is a piecewise linear function with slopes in $[-1,1]$. The graph of $h=\frac 12 b_\lambda(t)$ (see Figure \ref{Fig:bLambda}) gives the inner shape of the skew plane partition from Figure \ref{Fig:3d_skew_part}. 

\begin{figure}[ht]
\caption{\label{Fig:bLambda} The graph of $\frac 12 b_\lambda(t)$ when $\gl=\{4,4,3,3,3,1\}$.}
\includegraphics[width=6cm]{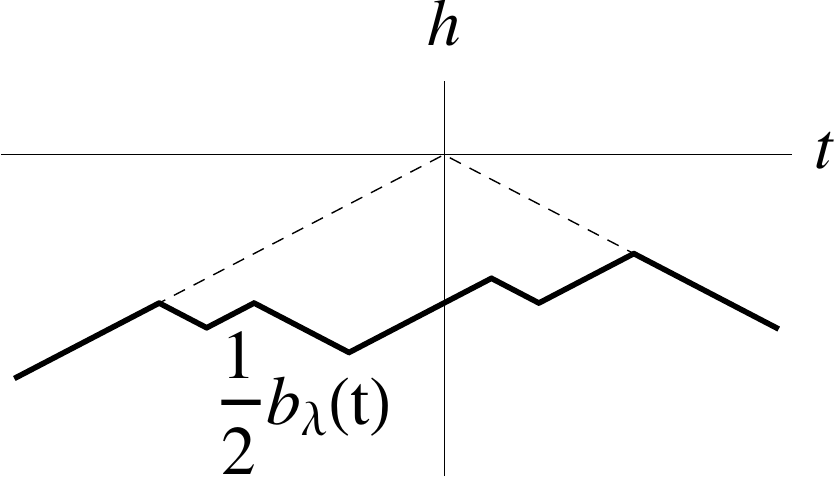}
\end{figure}

\subsubsection{The thermodynamical limit; cases studied before}
\label{subsec:before}
For $q\in (0,1)$ introduce a probability measure on skew plane partitions with boundary $\lambda$ and confined to a $c\times d$ box by 
$$Prob(\pi)\propto q^{|\pi|}.$$

Given a subset $U=\{(t_1,h_1),\ldots,(t_n,h_n)\}\subset \mathbb{Z}\times \frac 12 \mathbb{Z}$, define the corresponding local correlation function $\rho_{\lambda,q}(U)$ as the probability for a random tiling taken from the above probability space to have horizontal tiles centered at all positions $(t_i,h_i)_{i=1}^n$. 

Okounkov and Reshetikhin studied the thermodynamical limit $q=e^{-r} \rightarrow 1$ of this system in several cases. In \cite{OR1} they studied the case when $\gl$ is the empty partition and the size of the box is infinite (i.e. plane partitions are not confined to a finite $c\times d$ box). 

In \cite{OR2} they showed that in general, for arbitrary $\gl$, the correlation functions are determinants.
\begin{theorem}[Theorem 2, part 3 \cite{OR2}] 
\label{thm:fin-corr2}
The correlation functions $\rho_{\lambda,q}$ are determinants
$$\rho_{\lambda,q}(U)=\det(K_{\lambda,q}((t_i,h_i),(t_j,h_j)))_{1\leq i,j\leq n},$$ where the correlation kernel $K_{\gl,q}$ is given by the double integral
\begin{multline}\label{eq:main-corr2}
K_{\gl,q}((t_1,h_1),(t_2,h_2))=\\ \frac{1}{(2\pi \ii)^2}
\int_{z\in C_z}\int_{w\in C_w}
\frac{\Phi_{b_\lambda}(z,t_1)}{\Phi_{b_\lambda}(w,t_2)}
\frac{\sqrt{zw}}{z-w}z^{-h_1+\frac 12 b_\gl(t_1)-\frac 12}w^{h_2-\frac 12 b_\gl(t_2)+\frac12}\frac{dzdw}{zw},
\end{multline}
where $b_\lambda(t)$ is the function giving the back wall corresponding to $\lambda$ as in Figure \ref{Fig:bLambda}, 
\begin{equation} \label{eq:Phis}
\begin{aligned}
&\Phi_{b_\lambda}(z,t)&=&\frac{\Phi_{-,b_\lambda}(z,t)}{\Phi_{+,b_\lambda}(z,t)},\\
&\Phi_{+,b_\lambda}(z,t)&=&\prod_{m>t, m\in D^+, m\in \ZZZ+\frac 12}(1-zq^m),\\
&\Phi_{-,b_\lambda}(z,t)&=&\prod_{m<t, m\in D^-, m\in \ZZZ+\frac 12}(1-z^{-1}q^{-m}),
\end{aligned}
\end{equation}
$m\in D^{\pm}$ means the back wall at $t=m$, i.e. $b_\lambda(t)$ at $t=m$ has slope $\mp 1$, and $C_z$ (resp. $C_w$) is a simple positively oriented contour around 0 such that its interior contains none of the poles of $\Phi_{b_\lambda}(\cdot,t_1)$ (resp. of $\Phi_{b_\lambda}(\cdot,t_2)^{-1}$). Moreover, if $t_1< t_2$, then $C_z$ is contained in the interior of $C_w$, and otherwise, $C_w$ is contained in the interior of $C_z$. 
\end{theorem}

Using Theorem \ref{thm:fin-corr2} they studied the limit $q\rightarrow 1$ of the statistical system when the number of corners in the partition $\gl$ stays finite and showed that it exhibits the limit shape phenomenon. More precisely, let $\{r_k\}_{k=1}^{\infty}$ be a positive sequence which converges to $0$. Set $q_k=e^{-r_k}$. Let $\gl_k(t)$ be a collection of partitions confined to $c_k\times d_k$ boxes with corners at positions $u^{\lambda_k}_1\leq u^{\lambda_k}_2\leq\ldots\leq u^{\lambda_k}_{n-1}$, and let $b_{\lambda_k}(t)$ be the corresponding functions giving the back walls. Note, that the positions of corners depend on $k$ but the number of corners does not. Consider the probability measure on skew plane partitions $\pi^k$ with boundary $\lambda_k$ and confined to a $c_k\times d_k$ box by 
\begin{equation}
\label{eq:distr_k}
Prob(\pi^k)\propto q_k^{|\pi^k|}. 
\end{equation}

The characteristic scale is $r_k^{-1}$, so the system should be scaled by $r_k$. Let $B_{\lambda_k}(\tau)$ denote the function corresponding to $b_{\lambda_k}(t)$ when scaled by $r_k$, i.e. $B_{\lambda_k}(\tau):=r_k b_{\lambda_k}(\frac{\tau}{r_k}).$ Given constants $V_0<V_1<\ldots <V_{n-1}<V_n$ such that $V_0<0<V_n$, Okounkov and Reshetikhin studied the limit $k\rightarrow\infty$ when $r_k\rightarrow 0$, $r_k u^{\lambda_k}_i\rightarrow V_i$ for all $1\leq i \leq n-1$, $\frac 1{\sqrt{2}} r_k c_k\rightarrow -V_0$, and $\frac 1{\sqrt{2}} r_k d_k\rightarrow V_n$. In other words, in the scaling limit the back wall is a piecewise linear curve with line segments $V_i<\tau<V_{i+1}$ of slopes $\pm 1$ (see Figure \ref{Fig:OR2-lambda}).
\begin{figure}[ht]
\caption{\label{Fig:OR2-lambda}The back wall in \cite{OR2}.}
\includegraphics[width=6cm]{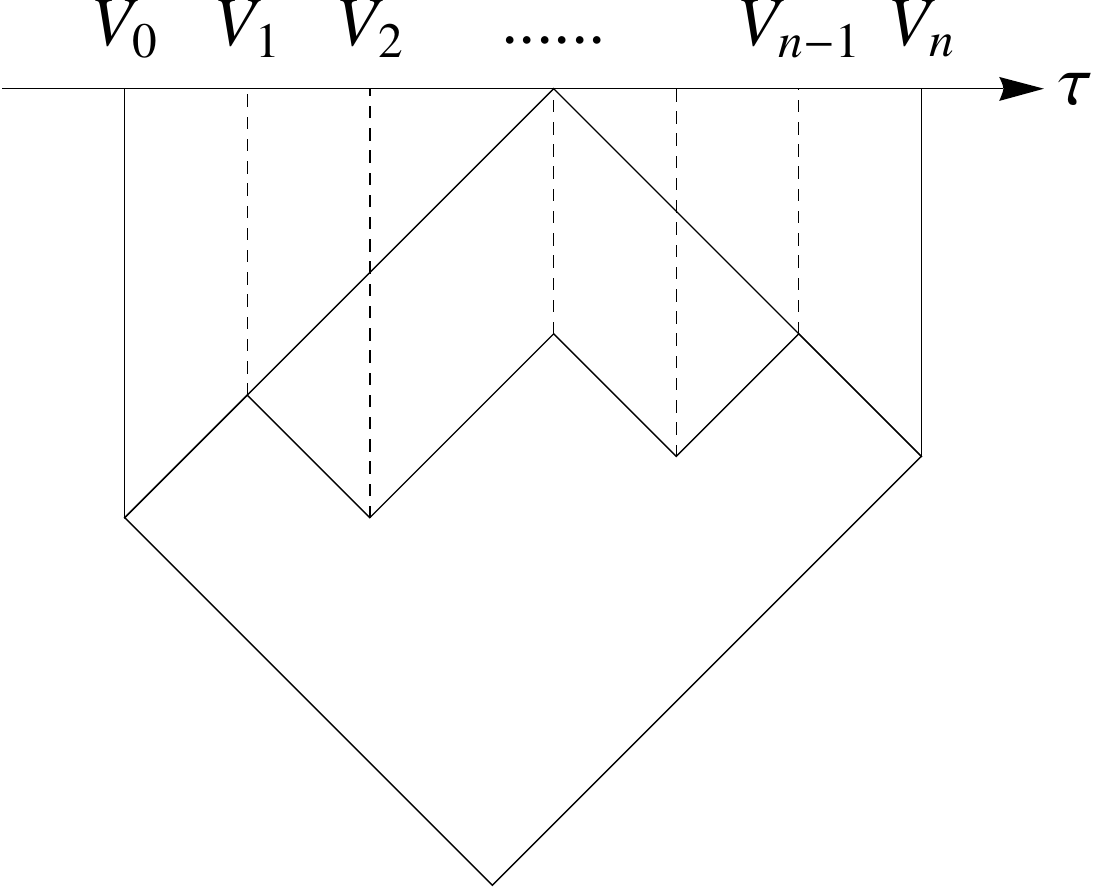}
\end{figure}

They showed that the scaled height function $r_k\pi^k$ of a random skew plane partition converges to a deterministic shape. In the two dimensional formulation of the system as random tilings of the plane by lozenges two different types of regions appear: those where only one type of tiles appear, called frozen regions, and those, where tiles of each type appear with positive density, called liquid regions. They showed that the correlation functions in the limit are described by a determinantal process given by the incomplete Beta kernel in the liquid region. They gave a characterization of the boundary between the frozen and liquid phases that arise and showed that in the limit the correlation kernel $K$ is given by the Airy kernel at generic points of the frozen boundary and by the Piercey process at the singular points of the frozen boundary, which appear near the outer corners (see Figure \ref{Fig:OR2-frozen-boundary}).

\begin{figure}[ht]
\caption{\label{Fig:OR2-frozen-boundary} The frozen boundary in \cite{OR2}.}
\includegraphics[width=8cm]{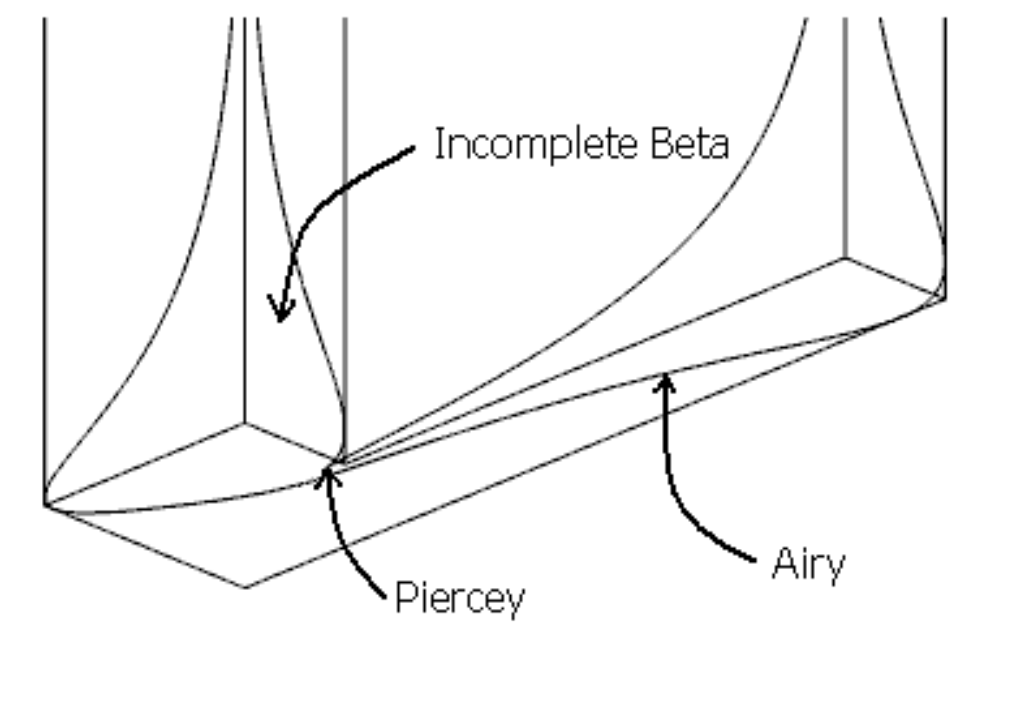}
\end{figure}

While in \cite{OR1} and \cite{OR2} the limiting back wall consisted of a piecewise linear function of slopes $\pm 1$, in \cite{BMRT} we studied the case when the limiting back wall is piecewise periodic in such a way, that in the scaling limit it converges to a continuous piecewise linear function of rational slopes strictly between $-1$ and $1$ (see Figure \ref{Fig:BMRT-periodic-wall}).  We showed that the correlation functions in the limit are given by the same processes as before, both in the bulk and at generic points of the \begin{figure}[ht]
\caption{\label{Fig:BMRT-periodic-wall} The periodic back wall in \cite{BMRT}.}
\includegraphics[width=14cm]{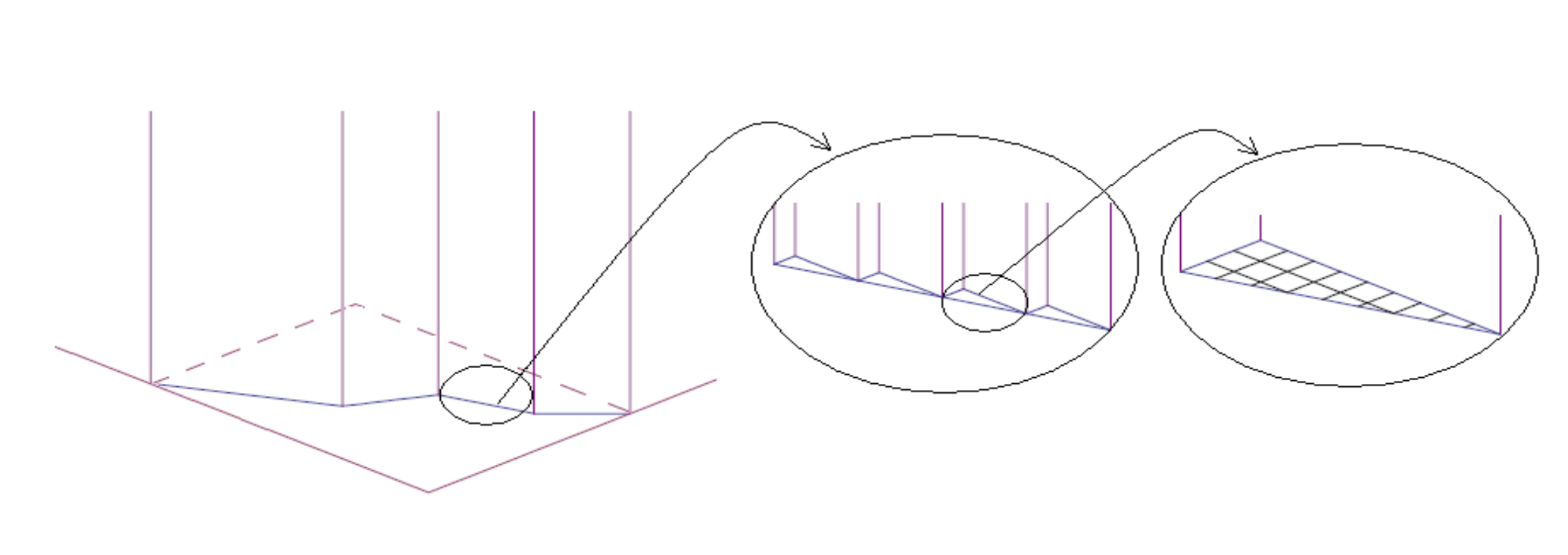}
\end{figure}
frozen boundary. Unlike the \cite{OR2} case, the frozen boundary doesn't develop cusps (see Figure \ref{Fig:BMRT-frozen}); hence, the Piercey process does not appear. Another difference is that in this setup the liquid region extends to $\infty$ everywhere on the back wall (see Figure \ref{Fig:BMRT-frozen}), and the local statistics is given by the bead process of Boutillier \cite{Bou} when you move high up on the wall. The bead process also appears high up, near the corners, with various parameters depending on the slope of the approach to the corner. 
\begin{figure}[ht]
\caption{\label{Fig:BMRT-frozen} The frozen boundary in \cite{BMRT}.}
\includegraphics[width=7cm]{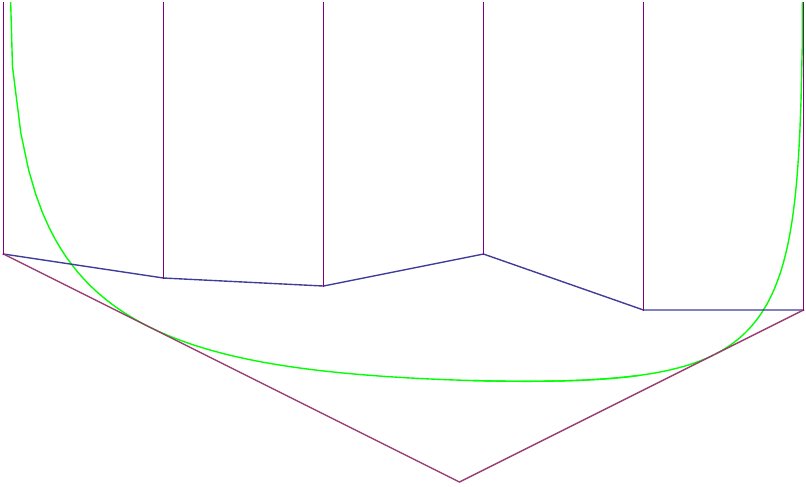}
\end{figure}

\subsection{Main results}

The present work studies the general case when the limiting back wall is an arbitrary continuous piecewise linear function of slopes in $[-1,1]$. This is a generalization of \cite{OR2} where slopes were $\pm 1$ (lattice slopes) and of \cite{BMRT} where slopes were in $(-1,1)$ (non-lattice slopes). 

Let $r_k,q_k,c_k,d_k,\lambda_k,b_{\lambda_k}$ and $B_{\lambda_k}$ be defined as in Section \ref{subsec:before}. In Section \ref{subsec:before} the number of corners in $\lambda_k$ was independent of $k$, however here there is no restriction. Consider the limit $k\rightarrow\infty$ when $r_k\rightarrow 0$, $r_k c_k\rightarrow C$, $r_k d_k\rightarrow D$, and the curves $B_{\lambda_k}(\tau)$ converge point-wise and uniformly to a continuous piecewise linear function $V(\tau)$ with slopes in $[-1,1]$. The following notation will be used throughout the paper. Let $V_0<\ldots<V_{j_\tau-1}<\gt<V_{j_\tau}<\ldots <V_n$ and $\gb_1,\ldots,\gb_n\in[-1,1]$ be such that $V(\gt)$ is linear for $\gt\in[V_{i-1},V_i]$ with slope $\gb_i$. Note that $j_\tau$ will denote the index where $V_{j_\tau-1}<\gt<V_{j_\tau}$. In order to simplify notation, from now on the index $\tau$ from $j_\tau$ will be dropped. For convenience also define $\gb_0=-1$ and $\gb_{n+1}=1$, and require that $\gb_i\neq\gb_{i+1}$ for all $i$.

\begin{figure}[ht]
\caption{\label{Fig:M-Bdry}A frozen boundary.}
\includegraphics[width=7cm]{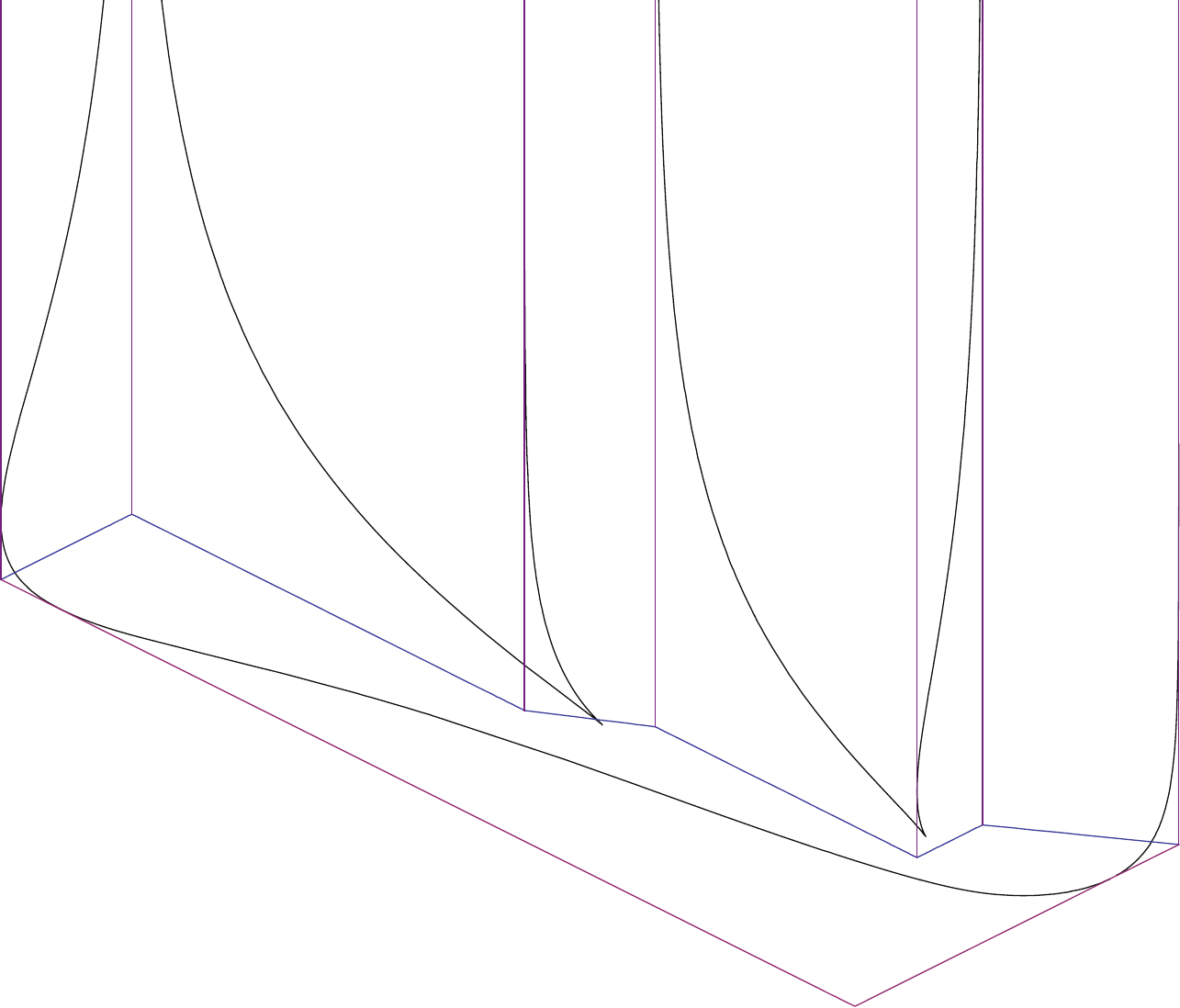}
\end{figure}
In such a scaling limit the statistical system of random skew plane partitions distributed according to \eqref{eq:distr_k} behaves as follows:
\begin{enumerate}
\item The system exhibits the limit shape phenomenon.
\item Near sections of the back wall where the slope is $\pm 1$ the disordered region is bounded above as in \cite{OR2}. Near sections with slope in $(-1,1)$ it grows infinitely high as in \cite{BMRT} (see Figure \ref{Fig:M-Bdry}). 
\item The number of connected components of the frozen boundary is one more than the number of outer corners where at least one of the slopes at the corner is a lattice slope. Note: outer corners are those where $V'(\tau-)<V'(\tau+)$.
\item The frozen boundary develops a cusp for each such outer corner.
\item The correlation functions are given by determinants with the incomplete beta kernel in the bulk, the Airy kernel on the frozen boundary, and the Piercey kernel near the cusps, as in \cite{OR2}.
\item High up, near the sections of the back wall which have non-lattice slopes, the correlation functions converge to the bead process as in \cite{BMRT}. High up, near the sections of the back wall which have lattice slopes, the region is frozen.
\item When approaching a corner and simultaneously going up along the frozen boundary, depending on the relative speed of the approach and the slopes of the back wall at the corner, either a frozen region or the bead process can be seen, with the density of the beads depending on the relative speed of the approach. 
\end{enumerate}

In this case, unlike \cite{OR1},\cite{OR3},\cite{OR2} and \cite{KO}, the boundary of the limit shape is not an algebraic curve.

Cusps on the frozen boundary can exhibit previously unknown behavior, which will be studied in future articles. One such example is when cusps approach each other as shown in Figure \ref{Fig:CuspsPass}. 

\begin{figure}[ht]
\caption{\label{Fig:CuspsPass}The frozen boundary when the back wall is given by $V = {-12.1, -12, -8.1, -8, -4, 0}$ and $\gb = {-0.9, -1, -0.9, 1, -1}$. The left picture is the right picture zoomed in inside the circle.}
\includegraphics[width=15cm]{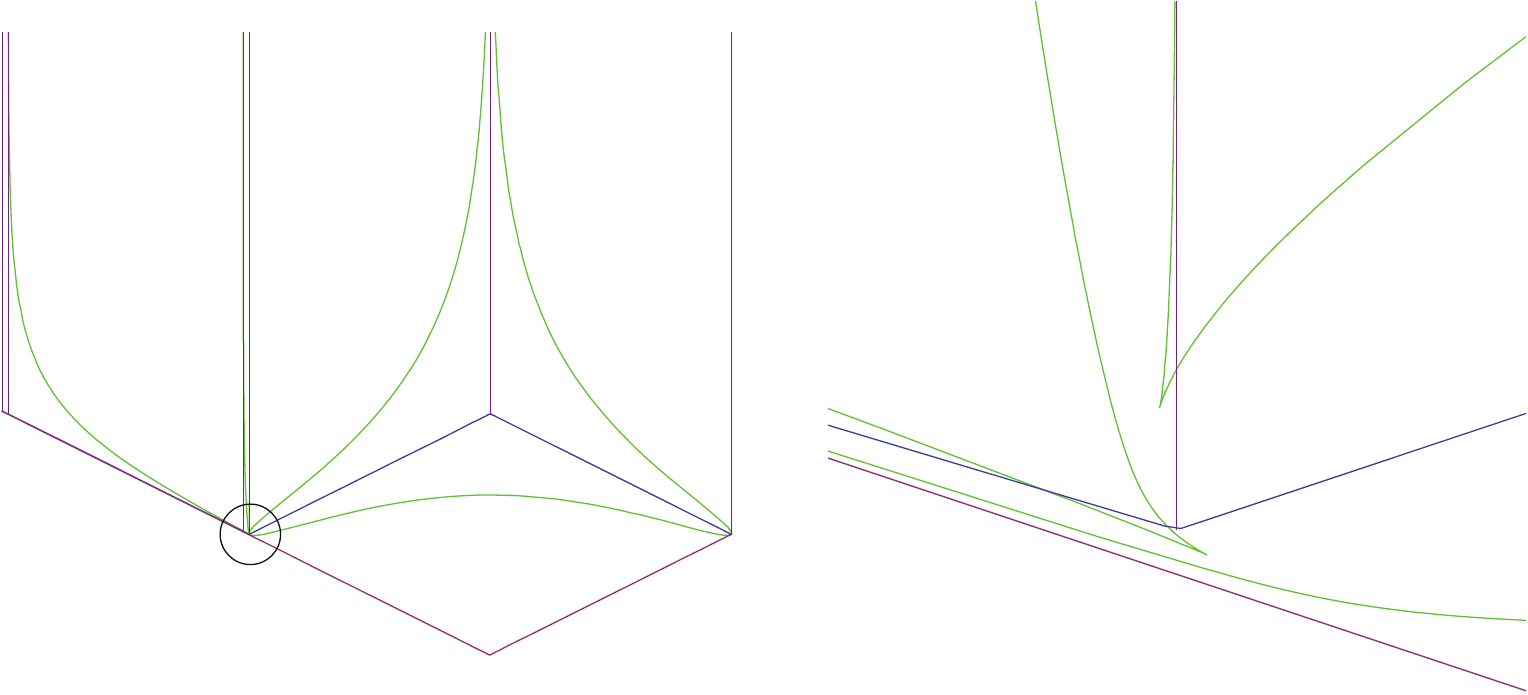}
\end{figure}

It is also shown in this paper that the local statistics in the scaling limit is independent of the intermediate approximations. Let $r_k$ and $q_k$ be as before. Let $t^k_1\in\ZZZ$, $t^k_2\in\ZZZ$, $h^k_1\in\frac 12\ZZZ$, and $h^k_2\in\frac 12\ZZZ$ be sequences such that
\begin{equation}
\label{eq:th}
\lim_{k\rightarrow\infty}t^k_1=\lim_{k\rightarrow\infty}t^k_2=\tau\text{ and }\lim_{k\rightarrow\infty}h^k_1=\lim_{k\rightarrow\infty}h^k_2=\chi.
\end{equation}
The following theorem is proven.
\begin{theorem} \label{thm:IndOfFam} 
Let $\lambda_k(t)$ and $\mu_k(t)$ be two sequences of partitions with corresponding scaled boundary functions $B_{\lambda_k}(\tau)$ and $A_{\lambda_k}(\tau)$. If in the limit $k\rightarrow \infty$ both $B_{\lambda_k}(\tau)$ and $A_{\lambda_k}(\tau)$ converge uniformly to the same piecewise linear function $V(\tau)$ corresponding to the back wall, then 
$$\lim_{k\rightarrow \infty}K_{\gl_k,q_k}((t^k_1,h^k_1),(t^k_2,h^k_2))= \lim_{k\rightarrow \infty} K_{\gm_k,q_k}((t^k_1,h^k_1),(t^k_2,h^k_2)).$$
\end{theorem}

In other words, $\lim_{k\rightarrow \infty}K_{\gl_k,q_k}((t^k_1,h^k_1),(t^k_2,h^k_2))$ only depends on the limiting back wall $V(\tau)$. It does not depend on the intermediate steps $\lambda_k(t)$. In particular this is true for the one point function $\lim_{k\rightarrow\infty}\rho_{\lambda_k,q_k}(t^k_1,h^k_1)=\lim_{k\rightarrow \infty}K_{\gl_k,q_k}((t^k_1,h^k_1),(t^k_1,h^k_1))$, which gives the slope of the limit shape at the point $(\tau,\chi)$. Thus, it is obtained that the limit shape is also independent of the intermediate steps.

\subsection{Outline of the structure of the paper}
One of the main steps in what follows is to understand the asymptotics of $\ln \Phi_{b_{\lambda_k}}$. This is done in Appendix \ref{sec:S_appendix}. The reason for separating this from the rest of the paper is that the argument actually holds in greater generality; one can consider back walls where $V(\tau)$ is an arbitrary continuous Lipshitz function with constant 1.

Section \ref{sec:IndepOfSeq} studies the critical points of the asymptotically leading term in the integral formula giving the correlation kernel. The main result is that the number of non-real complex critical points is $0$ or $2$. The general strategy for proving this is the same as that from \cite{BMRT}; obtain the number of non-real complex critical points when $\chi\rightarrow\infty$ and from this deduce the result for finite $\chi$. However, the argument in \cite{BMRT} could not be used here, since it heavily relies on the frozen boundary having a certain simple shape with only one connected component (see Figure \ref{Fig:BMRT-frozen}). 

The results of Section \ref{sec:IndepOfSeq} are used in Section \ref{sec:CorrKer} to apply the saddle point method and obtain the limit of the correlation kernel in the bulk. 

Section \ref{sec:TheBoundary} gives a detailed study of the frozen boundary. Such a study has not been carried out before.

\subsection{Acknowledgements}
I am very grateful to Peter Tingley for a suggestion on which the proof in Appendix \ref{sec:S_appendix} is based. I am very grateful to Cedric Boutillier and to an anonymous referee for many suggestions to improve the presentation of this paper. I am very grateful to Nicolai Reshetikhin for his guidance. I am also very grateful to Alexei Borodin, Cedric Boutilier and Peter Tingley for many useful discussions on the subject. Lastly, I would like to thank the organizers of the Park City Mathematics Institute Summer School on statistical mechanics, in July 2007, where much of Appendix \ref{sec:S_appendix} was written. 

\section{Critical points of the asymptotically leading term in the integral formula giving the correlation kernel}
\label{sec:IndepOfSeq}
\subsection{The function \texorpdfstring{$S(z)$}{S(z)}}
Lemma \ref{lem-IndOfFam} gives an asymptotical formula for $\Phi_{b_{\lambda_k}}(z,t)$ defined in \eqref{eq:Phis}, when the partitions are not restricted to a box. If they are bounded to a box in such a way that $\gt\in[V_0,V_n]$, then from the definition of $\Phi_{b_{\lambda_k}}(z,t)$ it is easy to see that if $\tilde b_{\lambda_k}(t)$ are defined by the conditions
$$\begin{array}{lcc}
\tilde b_{\lambda_k}(t)&\text{ is}&\text{continuous},\\
\tilde b'_{\lambda_k}(t)=&-1,&r_kt<V_0,\\
\tilde b_{\lambda_k}(t)=&b_{\lambda_k}(t),&V_0<r_kt<V_n,\\
\tilde b'_{\lambda_k}(t)=&1,&V_n<r_kt,
\end{array}$$
then $\Phi_{\tilde b_{\lambda_k}}(z,t)=\Phi_{b_{\lambda_k}}(z,t)$. Moreover, in the limit $r_k\rightarrow 0$, the sequence $\tilde B_{\lambda_k}(\tau)$ (i.e. the scaled $\tilde b_{\lambda_k}(t)$) converges to the function $\tilde V(\gt)$ defined by the conditions
$$\begin{array}{lcc}
\tilde V(\gt)&\text{ is}&\text{continuous},\\
\tilde V'(\gt)=&-1,&\gt<V_0,\\
\tilde V(\gt)=&V(\gt),&V_0<\gt<V_n,\\
\tilde V'(\gt)=&1,&V_n<\gt.
\end{array}$$

Let $t^k_i,h^k_i,i\in\{1,2\}$ be as in \eqref{eq:th}. The previous two statements, together with Corollary \ref{cor:limPhi} give that for plane partitions restricted to a box in such a way that $\gt\in[V_0,V_n]$,
\begin{multline*}
\lim_{k\rightarrow \infty}r_k\ln\Phi_{b_{\lambda_k}} (z,t^k_i)=
\lim_{k\rightarrow \infty}r_k\ln\Phi_{\tilde b_{\lambda_k}} (z,t^k_i)=\\
-\int_{V_0}^{\tau} -\frac{1}{2}(1+V'(M))\ln\(1-e^{M}z^{-1}\)dM +\int_\tau^{V_n}-\frac{1}{2}(1-V'(M)) \ln\(1-e^{-M}z\)dM.
\end{multline*}

Now, from the integral formula \eqref{eq:main-corr2} conclude that in the limit $r_k\rightarrow 0$, $r_kt^k_i \rightarrow\gt$, $r_kh^k_i\rightarrow\gq$, the correlation kernel can be written as
\begin{equation}
\label{eq:K-to-S}
K_{\lambda_k,q_k}((t^k_1,h^k_1),(t^k_2,h^k_2))=\frac{1}{(2\pi \ii)^2}\int\int e^{\frac{S_{\gt,\gq}(z)-S_{\gt,\gq}(w)}{r_k}+O(1)} \frac{1}{z-w}dwdz,
\end{equation}
where the contours of integration are as in \eqref{eq:main-corr2}, and $S_{\gt,\gq}(z)$ is defined as 
\begin{multline}
\label{eq:S_integral}
S_{\gt,\gq}(z):=-\int_{V_0}^{\tau} -\frac{1}{2}(1+V'(M))\ln\(1-e^{M}z^{-1}\)dM
\\+\int_\tau^{V_n}-\frac{1}{2}(1-V'(M))\ln\(1-e^{-M}z\)dM-\ln(z)\(\gq-\frac 12 V(\gt)\).
\end{multline}

\subsection{Number of complex critical points of \texorpdfstring{$S(z)$}{S(z)}}
\label{sec:2or0critPts}
In Section \ref{sec:CorrKer} the asymptotics of the correlation kernel will be studied by the saddle point method. Since the asymptotically leading term of the correlation kernel is given by the function $S_{\gt,\gq}(z)$, to use the saddle point method the critical points of $S_{\gt,\gq}(z)$ need to be studied. The goal of this section is to show that $S_{\gt,\gq}(z)$ has exactly 2 or 0 non-real complex critical points. The number of such critical points depends on the position of the point $(\gt,\gq)$. 

To simplify formulas, from now on the subscripts in $S_{\gt,\gq}(z)$ will be omitted unless that might cause ambiguities.

\subsubsection{Formulas for \texorpdfstring{$zS'$}{zS'} and \texorpdfstring{$z(zS')'$}{z(zS')'}}
In the study of critical points of $S(z)$ formulas for  $z\frac{dS(z)}{dz}$ and $z\frac{d}{dz}(z\frac{dS(z)}{dz})$ will be needed. The reason for working with these functions instead of $S'$ and $S''$ is that expressions for $S'$ and $S''$ are more complicated. 

It follows from \eqref{eq:S_integral} that
\begin{align*}
\nonumber
z\frac{dS(z)}{dz}=&\int_{V_0}^{\gt}\frac 12(1+V'(M))\frac{e^Mz^{-1}}{1-e^Mz^{-1}}dM+ \int_{\gt}^{V_n}\frac 12(1-V'(M))\frac{e^{-M}z}{1-e^{-M}z}dM
\\&-(\gq-\frac 12V(\gt))
\\=&-\sum_{i=1}^{j-1}\frac 12(1+\gb_i) \ln\(\frac{1-e^{V_i}z^{-1}}{1-e^{V_{i-1}}z^{-1}}\)
-\frac 12(1+\gb_j) \ln\(\frac{1-e^{\gt}z^{-1}}{1-e^{V_{j-1}}z^{-1}}\)
\\&+\frac 12(1-\gb_j) \ln\(\frac{1-e^{-V_j}z}{1-e^{-\gt}z}\)
+\sum_{i=j+1}^{n}\frac 12(1-\gb_i) \ln\(\frac{1-e^{-V_i}z}{1-e^{-V_{i-1}}z}\) \\&-(\gq-\frac 12V(\gt)).
\end{align*}

Using $V(V_i)=V(V_{i-1})+\gb_i(V_i-V_{i-1})$, $z\frac{dS(z)}{dz}$ can be rewritten as follows:
\begin{align}
z\frac{dS(z)}{dz}=
&\nonumber\label{eq:Sp}
-\sum_{i=1}^{j-1}\frac 12(1+\gb_i) \ln\(\frac{ze^{-V_i}-1}{ze^{-V_{i-1}}-1}\)
-\frac 12(1+\gb_j) \ln\(\frac{ze^{-\gt}-1}{ze^{-V_{j-1}}-1}\)
\\&+\frac 12(1-\gb_j) \ln\(\frac{ze^{-V_j}-1}{ze^{-\gt}-1}\)
+\sum_{i=j+1}^{n}\frac 12(1-\gb_i) \ln\(\frac{ze^{-V_i}-1}{ze^{-V_{i-1}}-1}\)
\\\nonumber &-\gq-\frac 12\gt+\frac 12(V(V_0)+V_0).
\end{align}

From here it is easy to obtain
\begin{equation}
\label{eq:Spp}
z\frac d{dz}\(z\frac d{dz} S(z)\)
=\frac 12(1+\gb_1)\frac 1{z-e^{V_0}}+\sum_{i=1}^{n-1}\frac 12(\gb_{i+1}-\gb_i)\frac 1{z-e^{V_i}}+\frac 12(1-\gb_n)\frac 1{z-e^{V_n}}-\frac 1{z-e^{\gt}}.
\end{equation}

\subsubsection{Critical points of $S(z)$ are away from $[e^{V_{i-1}},e^{V_{i}}]$ if $\gb_i\neq\pm 1$}
\label{subsec:ValidZ-s}
If $z_{cr}$ is a critical point of $S(z)$, then $z_{cr}\frac{dS(z_{cr})}{dz}=0$. 
\begin{lemma}
\label{lem:badIntervals}
If $i\neq j$ is such that $\gb_i\neq\pm 1$, then there does not exist $z_0\in[e^{V_{i-1}},e^{V_{i}}]$ such that $\lim_{z\rightarrow z_0}z\frac{dS(z)}{dz}=0$.
\end{lemma}

\begin{proof}
Suppose $z_0=x+\ii\varepsilon$, where $x\in(e^{V_{i-1}},e^{V_{i}})$ and $|\varepsilon|\ll1$. For such $z_0$ all terms in \eqref{eq:Sp}, except perhaps one, have imaginary parts close to zero. More precisely, $$\Im\left(z_0\frac{dS(z_0)}{dz}-c\ln\(\frac{z_0e^{-V_i}-1}{z_0e^{-V_{i-1}}-1}\)\right) = O(\varepsilon),$$
where $c$ is the coefficient of $\ln(\frac{z_0e^{-V_i}-1}{z_0e^{-V_{i-1}}-1})$ in \eqref{eq:Sp}. Since $\ln$ is the branch of the logarithm with an imaginary part in $(-\pi,\pi)$ with a cut along $\mathbb{R}_{-}$, then $$\Im\left(c\ln\(\frac{z_0e^{-V_i}-1}{z_0e^{-V_{i-1}}-1}\)\right) = \pm c \pi + O(\varepsilon),$$ which in turn implies $\Im(z_0\frac{dS(z_0)}{dz}) = \pm c \pi + O(\varepsilon)$. From \eqref{eq:Sp} it follows that $c=\frac 12(1+\gb_i)$ or $c=\frac 12(1-\gb_i)$. If $\gb_i\neq\pm 1$, then $c\neq 0$. Thus, there is no $z_0\in(e^{V_{i-1}},e^{V_{i}})$ such that $\lim_{z\rightarrow z_0}z\frac{dS(z)}{dz}=0$. The only remaining points are $z_0=e^{V_{i-1}}$ and $z_0=e^{V_i}$, but these are singular points of $z\frac{dS(z)}{dz}$ and therefore $\lim_{z\rightarrow z_0}z\frac{dS(z)}{dz}\neq0$.
$\qed$
\end{proof}

Notice, that if $\gb_i=1$ and $i\geq j$, or $\gb_i=-1$ and $i\leq j$, the coefficient $c$ is zero. In this case $\Im(z_0\frac{dS(z_0)}{dz}) = O(\varepsilon)$, and there may be critical points of $S(z)$ in the interval $(e^{V_{i-1}},e^{V_{i}})$. 

\subsubsection{Critical points when \texorpdfstring{$\gq\rightarrow\infty$}{chi is large}}
\label{subsec:NumberOfCriticalPointsAtInfty}

\begin{lemma}
\label{lem:largeChi}
Fix $\gt\in(V_{j-1},V_j)$. For sufficiently large $\gq$, $S_{\gt,\gq}(z)$ has no non-real complex critical points if $\gb_j=\pm 1$ and exactly two non-real complex critical points (which will be complex conjugates) if $\gb_j\neq\pm 1$.
\end{lemma}

Let us first prove the following lemma, which will be used in the proof of Lemma \ref{lem:largeChi}.
\begin{lemma}
\label{lem:geom}
Let $\{x_i\}_{i=1}^{m_1}$, $\{y_i\}_{i=1}^{m_2}$, $\{\xx_i\}_{i=1}^{m_1-1}$, $\{\yy_i\}_{i=1}^{m_2-1}$, $w$ and $\ww\neq 0$ be real numbers such that
$$x_{m_1}<\ldots<x_2<x_1<w<y_1<y_2<\ldots < y_{m_2}.$$
Let $\zeta$ be a complex number with $\Im{\zeta}\geq 0$. Define
\begin{equation*}
\begin{array}{c}
\mu_i=\mathrm{angle}(\zeta-x_{i+1},\zeta-x_i),\ \forall i=1,2,\ldots,m_1-1,\\
\nu_i=\mathrm{angle}(\zeta-y_i,\zeta-y_{i+1}),\ \forall i=1,2,\ldots,m_2-1,
\end{array}
\end{equation*}
and
$$
\alpha=\mathrm{angle}(\zeta-x_1,\zeta-w).
$$
\begin{figure}[ht]
\caption{\label{Fig:geom1} Setup of Lemma \ref{lem:geom}.}
\includegraphics[width=12cm]{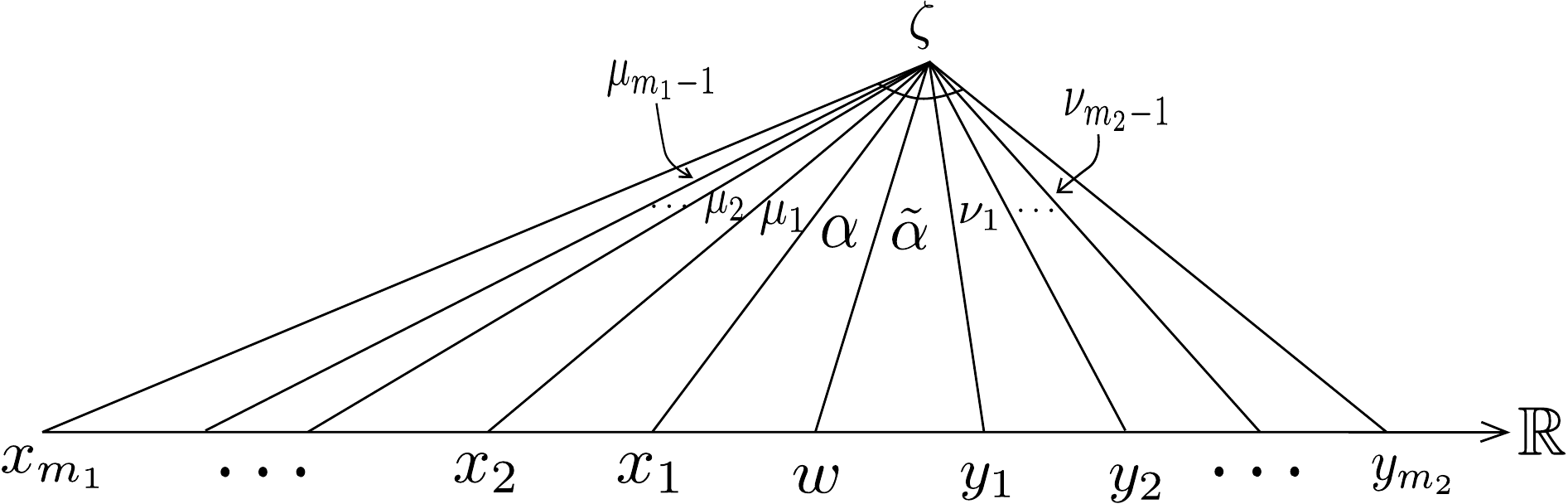}
\end{figure}
(See Figure \ref{Fig:geom1} for an illustration of the setup on the complex plane.) There exists $\varepsilon>0$ such that if $|\zeta-w|<\varepsilon$, then
\begin{equation}
\label{eq:angles}
\sum_{i=1}^{m_1-1} \xx_i\mu_i + \ww\alpha +\sum_{i=1}^{m_2-1} \yy_i\nu_i=0
\end{equation}
if and only if $\zeta\in\mathbb{R}$ and $\zeta>w$.

The same holds if $\alpha$ is replaced by $\tilde{\alpha}:=\mathrm{angle}(\zeta-w,\zeta-y_1)$ and $\zeta>w$ is replaced by $\zeta<w$.
\end{lemma}

\begin{proof}
Let $\varepsilon=y_1-w$. If $\zeta\in \RRR$, $\zeta>w$ and $|\zeta-w|<\varepsilon$, then $w<\zeta<y_1$. Thus, all the angles $\mu_i,\nu_i$ and $\alpha$ are zero, and \eqref{eq:angles} is true. 

Let us prove the converse. Suppose $\nexists \varepsilon>0$ such that \eqref{eq:angles} implies $\zeta\in\RRR$, $\zeta>w$. If $\zeta\in\RRR$ and $x_1<\zeta<w$, then all the angles $\mu_i$ and $\nu_i$ are zero, but $\alpha=\pi$ and \eqref{eq:angles} cannot hold. It follows, that there must be a sequence $\zeta_i\in\CCC\backslash\RRR$ such that 
\begin{equation}
\label{eq:limZeta}
\lim_{i\rightarrow\infty}\zeta_i=w
\end{equation}
and \eqref{eq:angles} holds for $\zeta=\zeta_i, \forall i$. Notice that the angles $\alpha,\mu_i$, and $\nu_i$ depend on $\zeta$ and that \eqref{eq:limZeta} implies 
\begin{equation}
\label{eq:limMuNu}
\lim_{i\rightarrow\infty}\mu_l(\zeta_i)=0\text{ and } \lim_{i\rightarrow\infty}\nu_l(\zeta_i)=0\text{ for all }l.
\end{equation}
These, together with the assumption that \eqref{eq:angles} holds for $\zeta=\zeta_i,\ \forall i$, imply 
\begin{equation}
\label{eq:limAlpha}
\lim_{i\rightarrow\infty}\alpha(\zeta_i)=0.
\end{equation}

Define $\eta(\zeta):=\mathrm{angle}(\zeta-w,\zeta-\Re \zeta)$ and $\gamma(\zeta):=\mathrm{angle}(\zeta-\Re \zeta,\zeta-y_1)$ (see Figure \ref{Fig:geom2}).
\begin{figure}[ht]
\caption{\label{Fig:geom2}}
\includegraphics[width=6cm]{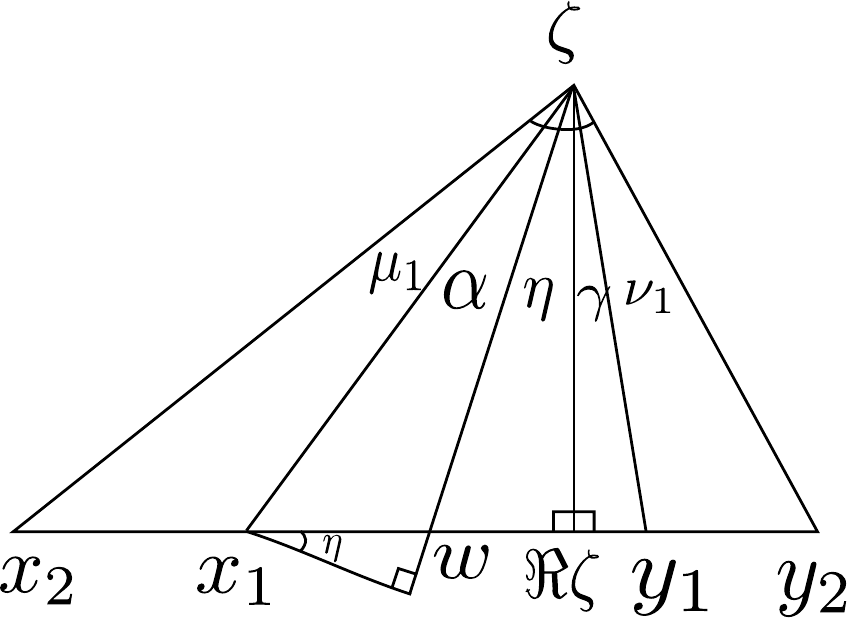}
\end{figure}
Using this notation,
$$\tan(\ga(\zeta_i))=\frac{|w-x_1|\cos(\eta(\zeta_i))}{|\zeta_i-w|+|w-x_1|\sin(\eta(\zeta_i))}.$$
Writing a similar expression for $\alpha(\zeta_i)+\mu_1(\zeta_i)$, and using 
\eqref{eq:limMuNu} and \eqref{eq:limAlpha}, obtain
\begin{multline}
\label{eq:geom-tan}
\lim_{i\rightarrow \infty}\frac{\gm_1(\zeta_i)+\ga(\zeta_i)}{\ga(\zeta_i)}
=\lim_{i\rightarrow \infty}\frac{\tan(\gm_1(\zeta_i)+\ga(\zeta_i))}{\tan(\ga(\zeta_i))}
=\\=\lim_{i\rightarrow \infty} \frac{\frac{|w-x_2|\cos(\eta(\zeta_i))}{|\zeta_i-w|+|w-x_2|\sin(\eta(\zeta_i))}}{
\frac{|w-x_1|\cos(\eta(\zeta_i))}{|\zeta_i-w|+|w-x_1|\sin(\eta(\zeta_i))}}
=\lim_{i\rightarrow \infty} \frac{\frac {|\zeta_i-w|}{|w-x_1|} + \sin(\eta(\zeta_i))}
{\frac {|\zeta_i-w|}{|w-x_2|}+\sin(\eta(\zeta_i))}.
\end{multline}

It follows from \eqref{eq:limZeta} and \eqref{eq:limAlpha} that  
\begin{equation}
\label{eq:limDelGam}
\lim_{i\rightarrow \infty}\eta(\zeta_i)=\lim_{i\rightarrow \infty}\gamma(\zeta_i)=\frac{\pi}{2}.
\end{equation}
Thus, \eqref{eq:geom-tan} gives that $\lim_{i\rightarrow \infty}\frac{\gm_1(\zeta_i)+\ga(\zeta_i)}{\ga(\zeta_i)}=1$, and hence that $\lim_{i\rightarrow \infty}\frac{\gm_1(\zeta_i)}{\ga(\zeta_i)}=0$. It is easy to see that using the same argument it can be shown that $\lim_{i\rightarrow \infty}\frac{\gm_l(\zeta_i)}{\ga(\zeta_i)}=0$ for all $l=1,2,3,\ldots,m_1$.

It follows from \eqref{eq:limZeta} and \eqref{eq:limDelGam}  that
\begin{multline*}
\lim_{i\rightarrow \infty}\frac{\gn_1(\zeta_i)}{\ga(\zeta_i)}
=\lim_{i\rightarrow \infty}\frac{\tan(\gn_1(\zeta_i))}{\tan(\ga(\zeta_i))}
=\lim_{i\rightarrow \infty} \frac{\frac{|y_1-y_2|\cos(\gamma(\zeta_i))}{|\zeta_i-y_1|+|y_1-y_2|\sin(\gamma(\zeta_i))}}{\frac{|w-x_1|\cos(\eta(\zeta_i))}{|\zeta_i-w|+|w-x_1|\sin(\eta(\zeta_i))}}
=\\=\frac{|y_1-y_2|}{|w-y_1|+|y_1-y_2|}\lim_{i\rightarrow \infty} \frac{\cos(\gamma(\zeta_i))}{\cos(\eta(\zeta_i))}
=\frac{|y_1-y_2|}{|w-y_1|+|y_1-y_2|}\lim_{i\rightarrow \infty} \frac{\sin(\frac \pi 2-\gamma(\zeta_i))}{\sin(\frac \pi 2-\eta(\zeta_i))}
=\\=\frac{|y_1-y_2|}{|w-y_1|+|y_1-y_2|}\lim_{i\rightarrow \infty} \frac{|\zeta_i-w|}{|\zeta_i-y_1|}=0.
\end{multline*}
Similarly, $\lim_{i\rightarrow \infty}\frac{\gn_l(\zeta_i)}{\ga(\zeta_i)}=0$ for all $l=1,2,3,\ldots,m_2$. Combining the results gives
$$\lim_{i\rightarrow\infty}\left(\sum_{l=1}^{m_1-1} \xx_l\frac{\mu_l (\zeta_i)}{\alpha(\zeta_i)} + \ww +\sum_{l=1}^{m_2-1} \yy_l\frac{\nu_l(\zeta_i)}{\alpha(\zeta_i)}\right)\neq0,$$
which is a contradiction to the assumption that \eqref{eq:angles} is satisfied for $\zeta=\zeta_i$ for all $i$.

This proves the first statement in the lemma. The second statement follows by symmetry.
$\qed$
\end{proof}

\begin{proof}[Proof of Lemma \ref{lem:largeChi}]
Since $z\neq 0$, studying the critical points of $S(z)$ is equivalent to studying  the solutions to $z\frac d{dz}S(z)=0$. The real part of the equation, i.e. $0=\Re(z\frac d{dz}S(z))$ implies that if $\gq$ is very large, then $z$ must be very close to $e^{\gt}$ or to $e^{V_l}$ for some $l$. 

The imaginary part $0=\Im(z\frac d{dz}S(z))$ is equivalent to
\begin{multline}
\label{eq:ImSp}
0 = 
-\sum_{m=1}^{j-1}
	\frac 12(1+\gb_m)
		\mathrm{angle}(z-e^{V_{m-1}},z-e^{V_m})
-\frac 12(1+\gb_j)\mathrm{angle}(z-e^{V_{j-1}},z-e^{\tau}) 
\\
+\frac 12(1-\gb_j)\mathrm{angle}(z-e^\tau,z-e^{V_j})
+\sum_{m=j+1}^{n}\frac 12(1-\gb_m)
	\mathrm{angle}(z-e^{V_{m-1}},z-e^{V_m}).
\end{multline}

For arbitrary real numbers $x,a,b$, if $x\notin[a,b]$, it is immediate that
$$\lim_{z\rightarrow e^x}\mathrm{angle}(z-e^a,z-e^b)=0.$$

This implies that if $z$ is near $e^{V_m}$, then all but the two terms in \eqref{eq:ImSp} where $e^{V_m}$ appears are close to zero. The sum of the angles in the remaining two terms is $\pi$, and both have coefficients of the same sign. If none of those two coefficients is zero, then if $z$ is sufficiently close to $e^{V_m}$, the RHS of \eqref{eq:ImSp} cannot be zero. 

If one of the two coefficients is zero, then in order for \eqref{eq:ImSp} to hold, $z$ must be real. This follows from Lemma \ref{lem:geom}. For example, in the case $\beta_m\neq -1$, $\beta_{m+1}=-1$ and $m<j-1$, setting $\zeta=z$, 
\begin{equation}
\label{eq:GeomLemCrspV}
\begin{array}{cccccccccccccccc}
x_{m_1}&x_{m_1-1}&\ldots&x_1&w&y_1&\ldots&y_{j-m-1}&y_{j-m}&y_{j-m+1}&\ldots&y_{m_2}\\
||&||&\ldots&||&||&||&\ldots&||&||&||&\ldots&||&\\
e^{V_0}&e^{V_1}&\ldots &e^{V_{m-1}}&e^{V_{m}}&e^{V_{m+1}}&\ldots&e^{V_{j-1}}&e^{\tau}&e^{V_{j}}&\ldots&e^{V_n}
\end{array},
\end{equation}
and
\begin{equation}
\label{eq:GeomLemCrspB}
\begin{array}{cccccccccccccc}
\xx_{m_1-1}&\ldots&\xx_1&\ww&\yy_1&\ldots&\yy_{j-m-1}&\yy_{j-m}&\ldots&\yy_{m_2-1}\\
||&\ldots&||&||&||&\ldots&||&||&\ldots&||&\\
-\frac{1+\beta_1}{2}&\ldots &-\frac{1+\beta_{m-1}}{2}&-\frac{1+\beta_{m}}{2}&-\frac{1+\beta_{m+2}}{2}&\ldots&-\frac {1+\beta_{j}}{2}&\frac{1-\beta_{j}}{2}&\ldots&\frac{1-\beta_{n}}{2},
\end{array}
\end{equation}
Lemma \ref{lem:geom} gives that \eqref{eq:ImSp} implies $z$ must be real.

If $z$ is near $e^\gt$ and $\gb_j=\pm 1$, then exactly one of the coefficients of angles in \eqref{eq:ImSp} containing $z-e^\gt$ is zero, and again it follows from Lemma \ref{lem:geom}, with the parameters set up similarly to \eqref{eq:GeomLemCrspV} and \eqref{eq:GeomLemCrspB} but this time with $w=e^\gt$, that $z$ must be real.

If $z$ is near $e^\gt$ and $\gb_j\neq\pm 1$, then using the methods from the proof of Proposition 3.1 of \cite{BMRT}, it can be shown that there are two possibilities for $z$, and these two complex conjugate critical points of $S(z)$ can be asymptotically calculated. Formulas for the critical points are given in Lemma \ref{lem:critPtsLargeChi}.
$\qed$
\end{proof}

\begin{lemma} In the limit $\chi\rightarrow\infty$
\label{lem:critPtsLargeChi}
\begin{enumerate}[(a)]
\item If $\gt\in(V_{j-1},V_j), \gb_j\neq\pm 1$ is fixed, then the asymptotics of the non-real complex critical points is given by $z_{cr}=e^{\gt-\varepsilon}$, where
\begin{align*}
\varepsilon = e^{-\chi} e^{\pm \ii\pi\frac{1}{2}(1+\beta_j)} 
\prod_{i=0}^n \left| 2 \sinh 
\left(\frac{\tau-V_i}{2}\right)\right|^{\frac 12 (\beta_{i+1}-\beta_i)}\left(1+O(e^{-\chi})\right).
\end{align*}
\item
If $\tau=V_{j-1}+\delta$, $\chi\to \infty $ and $\delta\to 0$ in such a way that
$$
p=e^{\chi-\chi^{(j-1)}}|\delta|^{1-\frac 12(\beta_j-\beta_{j-1})}
$$
is fixed, with
\begin{equation*}
e^{\chi^{(j-1)}} := \prod_{\substack{i=0 \\ i\neq j-1}}^{n} \Bigl| 2\sinh\bigl(\frac{V_{j-1}-V_i}{2}\bigr)\Bigr|^{\frac 12(\beta_{i+1}-\beta_i)},
\end{equation*}
then the critical points behave as $z_{cr}=e^{\tau-s|\delta|}$, where $s$ is a solution to the equation
\begin{equation}\label{eq:t-scaling}
  p=e^{\pm \ii\pi \frac{1}{2}(1+\beta_{j-1})} \frac{(s-\mathrm{sign}(\delta))^ {\frac 12 (\beta_j-\beta_{j-1})}}{s}.
\end{equation}
\end{enumerate}
\end{lemma}
\begin{proof}
Since the calculations are very similar to those in \cite{BMRT}, in order to avoid repetition, we will omit them here.
$\qed$
\end{proof}

\subsubsection{Nature of critical points in various limits}
\label{subsubsec:nearCorners}
Let us analyze the solutions to \eqref{eq:t-scaling}. Assume $\gd>0$. The case $\gd<0$ is similar. Consider the two limits $p\rightarrow\infty$ and $p\rightarrow 0$ in various scenarios depending on the angles $\gb_{j-1}$ and $\gb_j$. 

In the limit $p\rightarrow\infty$ the solution to \eqref{eq:t-scaling} has the asymptotics
\begin{equation}
\label{eq:s,pToInf}
s=p^{-1}e^{\pm \ii\pi\frac 12(1+\gb_j)}(1+O(p^{-1})).
\end{equation}

When $p\rightarrow 0$,
\begin{equation}
\label{eq:s,pTo0}
s=p^{-\frac 1{1-\frac 12(\gb_j-\gb_{j-1})}}e^{\pm \ii\pi \frac{\frac 12(1+\gb_{j-1})}{1-\frac 12(\gb_j-\gb_{j-1})}}(1+O(p)).
\end{equation}

To get these asymptotics, it is necessary to show that if $(\gb_j-\gb_{j-1})<0$ and $p\rightarrow\infty$, or $(\gb_j-\gb_{j-1})>0$ and $p\rightarrow 0$, then $s\nrightarrow 1$. This is a technicality that has been addressed in \cite{BMRT}.

In Lemma \ref{lem:geom} set 
\begin{align*}
\zeta&=z,\\
w&=e^\gt,\\
x_i&=e^{V_{j-i}}, i=1,2,\ldots,j,\\
y_i&=e^{V_{j-1+i}}, i=1,2,\ldots,n-j+1,\\
\xx_i&=-\frac 12(1+\beta_{j-i}), i=1,2,\ldots,j-1,\\
\yy_i&=\frac 12(1-\beta_{j+i}), i=1,2,\ldots,n-j. 
\end{align*}
Lemma \ref{lem:geom} and its proof can be used to show that if $z$ satisfies \eqref{eq:ImSp}, then in the following situations it must be real. The idea is to show that if $z$ is not real, then one of the angles in \eqref{eq:angles} with non-zero coefficient is much larger than all the other angles with non-zero coefficients, which is not possible.

\begin{itemize}
\item[Case 1:]$p\rightarrow\infty, s\rightarrow 0,\gb_j=1.$ In Lemma \ref{lem:geom} set $\ww=-1$. Then \eqref{eq:ImSp} has the form of \eqref{eq:angles}. Also,
$$|\zeta-w|=|e^{\gt-\varepsilon}-e^\gt|=e^\gt|\varepsilon|+o(\varepsilon)$$
and
$$|w-x_1|=|e^{V_{j-1}}-e^\gt|=e^\gt|\gd|+o(\gd).$$
Hence, $\frac{|\zeta-w|}{|w-x_1|}\approx\frac{|\varepsilon|}{\gd}=s\rightarrow 0$. The proof of Lemma \ref{lem:geom} implies that $\zeta$ must be real. Since $z=\zeta$, $z$ must be real.
\item[Case 2:]$p\rightarrow 0, s\rightarrow\infty,\gb_j=1.$ The setup is similar to the first case and $|\zeta-w|\rightarrow 0$, $|w-x_1|\rightarrow 0$. Unlike the previous case, $\frac{|\zeta-w|}{|w-x_1|} \rightarrow\infty$. The proof of Lemma \ref{lem:geom} gives that in the limit considered $\gm_1\gg\ga\gg(\text{all the other angles})$, which implies $\zeta$ must be real.
\item[Case 3:]$p\rightarrow 0, s\rightarrow\infty, \gb_{j-1}=-1, \gb_j\neq 1.$ In this case \eqref{eq:ImSp} has the form of \eqref{eq:angles} with the term $\ww \alpha$ in \eqref{eq:angles} replaced by $-\frac 12(1+\beta_j)\alpha+\frac 12(1-\beta_j)\tilde{\alpha}$. Notice that $\xx_1=0$. From the proof of Lemma \ref{lem:geom} it is easy to see that if $\zeta$ is not real, then $\mu_i\rightarrow 0$ for $i>1$, $\nu_i\rightarrow 0$ for $i\geq 1$, $\alpha+\tilde{\alpha}\gg\mu_i$ for $i>1$ and $\alpha+\tilde{\alpha}\gg\nu_i$ for $i\geq1$. Since $\frac{|\zeta-w|}{|w-x_1|}\approx\frac{|\varepsilon|}{\gd}=s\rightarrow \infty$, it follows that $\alpha\rightarrow 0$ (see Figure \ref{Fig:geom1}). Thus, $\tilde{\alpha}\rightarrow 0$ as well. A calculation similar to \eqref{eq:geom-tan} yields
$$\frac{\tan(\alpha)+\tan(\tilde{\alpha})}{\tan(\alpha)}\approx\frac{\frac{|\zeta-x_1|}{|w-x_1|}+1}{\frac{|\zeta-x_1|}{|y_1-x_1|}+1}\approx s\rightarrow\infty,$$
which implies $\tilde{\alpha}\gg\alpha$. Thus, in the considered limit $\tilde{\alpha}$ is much larger than all the other angles that appear with non-zero coefficients, which is impossible. Hence, $\zeta$ must be real.
\item[Case 4:]$p\rightarrow\infty, s\rightarrow 0, \gb_j=-1.$ Now \eqref{eq:ImSp} has the form of \eqref{eq:angles} if $\alpha$ is replaced by $\tilde{\alpha}$ and $\ww=1$. In this case $|\zeta-w|\rightarrow 0$, $|x_1-w|\rightarrow 0$, and $\frac{\zeta-w}{x_1-w}\rightarrow 0$. It is easy to see from the proof of Lemma \ref{lem:geom} that $\tilde{\ga}\gg\gm_i,\gn_i$ and hence that $\zeta$ must be real.
\end{itemize}

It is easy to see that the complex critical points $z_{cr}$ of $S(z)$ obtained from \eqref{eq:s,pToInf} and \eqref{eq:s,pTo0} are non-real complex except in the cases listed above.

\subsubsection{Critical points at finite $(\tau,\chi)$.}
\label{subsec:FollowOtherPointsToInfty}
In the previous section the number of non-real complex critical points of $S_{\tau,\chi}(z)$ was identified when $\gq$ is large. This section studies the number of non-real complex critical points for an arbitrary point $(\gt_0,\gq_0)$. Suppose that at this point the number of non-real complex critical points is neither 2 nor 0. It must be even, since they come in conjugate pairs. Assume there are $2m$ such critical points. The number of such critical points depends on the point $(\gt_0,\gq_0)$. However, if $(\gt_0,\gq_0)$ continuously changes in the $\gt,\gq$ plane, the number of non-real complex critical points of $S_{\tau_0,\chi_0}(z)$ will not change, until it reaches a point were the equations
\begin{equation}
\label{eq:Bdry}
S_{\gt,\gq}'(z)=S_{\gt,\gq}''(z)=0
\end{equation}
have a real solution $z\in\RRR$. In other words, the number of non-real complex critical points can change only near points $(\tau,\chi)$ where $S_{\tau,\chi}(z)$ has double real critical points. From \eqref{eq:Spp} and \eqref{eq:Sp} the condition \eqref{eq:Bdry} is equivalent to
\begin{align}
\label{eq:TauZ}
e^\gt=&z-\frac{1}{\sum_{i=0}^n\frac 12(\gb_{i+1}-\gb_i)\frac 1{z-e^{V_i}}},
\end{align}
and
\begin{align}
\nonumber
\gq=&-\sum_{i=1}^{j-1}\frac 12(1+\gb_i) \ln\(\frac{ze^{-V_i}-1}{ze^{-V_{i-1}}-1}\)
-\frac 12(1+\gb_j) \ln\(\frac{ze^{-\gt}-1}{ze^{-V_{j-1}}-1}\)
\\\label{eq:ChiZ}
&+\frac 12(1-\gb_j) \ln\(\frac{ze^{-V_j}-1}{ze^{-\gt}-1}\)
+\sum_{i=j+1}^{n}\frac 12(1-\gb_i) \ln\(\frac{ze^{-V_i}-1}{ze^{-V_{i-1}}-1}\)
\\\nonumber&-\frac 12\gt+\frac 12(V(V_0)+V_0).
\end{align}

Think of this as a curve $(\gt(z),\gq(z))$ in the $\gt,\gq$ plane parametrized by $z\in\RRR$. Consider the complement of this curve in the $\gt,\gq$ plane. The number of non-real complex critical points of $S_{\gt,\gq}$ is the same for all points $(\gt,\gq)$ inside the same connected component of this complement. Let $\DD$ be the connected component which contains $(\gt_0,\gq_0)$ and let $(\gt(z_0),\gq(z_0))$ be a generic point on the boundary of $\DD$. The complement of $(\gt(z),\gq(z))$ has another connected component whose boundary contains $(\gt(z_0),\gq(z_0))$. Call it $\tilde{\DD}$. The number of non-real complex critical points of $S_{\gt,\gq}$ is $2m$ for points $(\gt,\gq)\in\DD$ and $2m\pm 2$ for points $(\gt,\gq)\in\tilde{\DD}$. From Lemma \ref{lem:badIntervals} it follows that $z_0$ must be in one of the intervals $(e^{V_{i-1}},e^{V_i}),\gb_i=\pm 1$ or in $(-\infty,e^{V_0})\union(e^{V_n},\infty)$. Suppose, for example, $z_0\in(e^{V_{l-1}},e^{V_l})$. The portion of the curve $(\gt(z),\gq(z))$ corresponding to $z\in(e^{V_{l-1}},e^{V_l})$ is contained in $\overline{\DD}\cap\overline{\tilde{\DD}}$, where $\overline{\DD}$ denotes the closure of $\DD$ (a detailed study of the curve $(\gt(z),\gq(z))$ is carried out in Section \ref{sec:TheBoundary}). Now, when $z$ approaches $e^{V_l}$ (or $e^{V_{l-1}}$), $\gq(z)$ will approach $\infty$. However, it was shown in the previous section, that when $\gq$ is very large, the number of complex critical points is either 2 or 0. Thus, $2m=0$ or $2m=2$.

This establishes the following proposition:
\begin{proposition}
\label{prop:NumOfCritPtsEverywhere}
For any $(\gt,\gq)$ the number of non-real complex critical points of $S_{\gt,\gq}(z)$ is $2$ or $0$. Divide the $\gt,\gq$ plane into regions  according to the number of non-real complex critical points of $S_{\gt,\gq}(z)$. Each connected component where the number of such critical points is $0$ has points where $\gq$ is arbitrarily large.
\end{proposition}

\section{Asymptotics of the correlation kernel}
\label{sec:CorrKer}

This section analyzes the asymptotics of the correlation kernel $K_{\lambda_k,r_k}((t_1^k,h_1^k),(t_2^k,h_2^k))$ in the scaling limit when $\lim_{k\rightarrow \infty}r_k=0$, the skew plane partitions are scaled by $r_k$ in all directions, $$\lim_{k\rightarrow \infty}r_kt_1^k=\lim_{k\rightarrow \infty}r_kt_2^k=\gt,\ \lim_{k\rightarrow \infty}r_kh_1^k=\lim_{k\rightarrow \infty}r_kh_2^k=\gq,$$ and $\Delta(t):=t_1^k-t_2^k$ and $\Delta(h):=h_1^k-h_2^k$ are constants. A version of the saddle point method is used for calculating the asymptotics of the correlation kernel \eqref{eq:main-corr2}. The arguments used are along the lines of \cite{OR1},\cite{OR2},\cite{BMRT}.

Let $\gt,\gq$ be such that $S_{\gt,\gq}$ has two non-real complex critical points. Deform the contours of integration $C_z$ and $C_w$ in the double integral representation of the correlation kernel given in \eqref{eq:main-corr2} to $C'_z,C'_w$ in such a way that the new contours pass transversely through the two critical points of $S(z)$ and $\Re(S(z))\leq\Re(S(w))$ $\forall z\in C'_z,\forall w\in C'_w$, with equality if and only if $z=w=z_{cr}$. It was shown in the proof of Theorem 4.1 in \cite{BMRT} that this can be done for the function $S(z)$ when $\beta_i\neq\pm 1, \forall i$. The argument used there is general and applies for arbitrary $\beta_i\in[-1,1]$. During this contour deformation the contours cross each other along a path between the complex critical points of $S(z)$, so the residues from the term $\frac 1{z-w}$ should be picked. Thus, the integral \eqref{eq:main-corr2} can be written as 
\begin{multline*}
K_{\lambda_k,q_k}((t_1^k,h_1^k),(t_2^k,h_2^k)) = \frac{1}{(2\pi \ii)^2}
\int_{z\in C'_z}\int_{w\in C'_w}
\frac{\Phi_{-,b_{\lambda_k}}(z,t_1^k)\Phi_{+,b_{\lambda_k}}(w,t_2^k)}{\Phi_{+,b_{\lambda_k}}(z,t_1^k)\Phi_{-,b_{\lambda_k}}(w,t_2^k)}\times
\\\times \frac{\sqrt{zw}}{z-w}z^{-h_1^k+\frac 12 b_{\lambda_k}(t_1^k)-1/2} w^{h_2^k-\frac 12 b_{\lambda_k}(t_2^k)+1/2}\frac{dzdw}{zw} 
\\+ \frac{1}{2\pi \ii}\int_{z_{cr_1}}^{z_{cr_2}} 
\frac{\Phi_{-,b_{\lambda_k}}(z,t_1^k)\Phi_{+,b_{\lambda_k}}(z,t_2^k)}{\Phi_{+,b_{\lambda_k}}(z,t_1^k)\Phi_{-,b_{\lambda_k}}(z,t_2^k)}
z^{h_2^k-h_1^k+\frac 12 b_{\lambda_k}(t_1^k)-\frac 12 b_{\lambda_k}(t_2^k)-1}dz.
\end{multline*}

Recall that the first integral in the above formula has the form
\begin{equation*}
\frac{1}{(2\pi \ii)^2}\int_{C'_z}\int_{C'_w} e^{\frac{S_{\gt,\gq}(z)-S_{\gt,\gq}(w)}{r_k}+O(1)} \frac{1}{z-w}dwdz.
\end{equation*}
In the limit $k\rightarrow \infty$ the asymptotically leading term of the double integral is $e^{\frac{S(z)-S(w)}{r_k}}$ and $\Re(S(z))<\Re(S(w))$ along the contours $C'_z,C'_w$ except at the critical points. This implies that the main contribution to the integral comes from the critical points, but since the contours of integration cross transversely at the critical points, the integral is zero in the limit. Hence,
\begin{multline*}
\lim_{k\rightarrow \infty} K_{\lambda_k,q_k}((t_1^k,h_1^k),(t_2^k,h_2^k)) =\\= \lim_{k\rightarrow \infty}\frac{1}{2\pi \ii}\int_{z_{cr_1}}^{z_{cr_2}} 
\frac{\Phi_-(z,t_1^k)\Phi_+(z,t_2^k)}{\Phi_+(z,t_1^k)\Phi_-(z,t_2^k)}
z^{h_2^k-h_1^k+\frac 12 b_{\lambda_k}(t_1^k)-\frac 12 b_{\lambda_k}(t_2^k)-1}dz.
\end{multline*}

Suppose $t_1^k<t_2^k$. Let $P_{\gl_k}(t_1^k,t_2^k)=\#\{D^-\cap(t_1^k,t_2^k)\}$. The correlation kernel can be written as follows:
\begin{align*}
\lim_{k\rightarrow \infty} & K_{\lambda_k,q_k}((t_1^k,h_1^k),(t_2^k,h_2^k)) =
\\=& \lim_{k\rightarrow \infty} \frac{1}{2\pi \ii}\int_{z_{cr_1}}^{z_{cr_2}} 
\frac {z^{h_2^k-h_1^k+\frac 12 b_{\lambda_k}(t_1^k)-\frac 12 b_{\lambda_k}(t_2^k)-1}} {\prod_{\stackrel{t_1^k<m<t_2^k}{b_{\lambda_k}'(m)=1}}(1-e^{r_km}z^{-1}) \prod_{\stackrel{t_1^k<m<t_2^k}{b_{\lambda_k}'(m)=-1}}(1-e^{-r_km}z)} dz 
\\=&
\lim_{k\rightarrow \infty} \frac{1}{2\pi \ii}\int_{z_{cr_1}}^{z_{cr_2}} \prod_{\stackrel{t_1^k<m<t_2^k}{b_{\lambda_k}'(m)=1}}(-e^{-r_km}z) \prod_{t_1^k<m<t_2^k}\frac 1{(1-e^{-r_km}z)}z^{h_2^k-h_1^k+\frac 12 b_{\lambda_k}(t_1^k)-\frac 12 b_{\lambda_k}(t_2^k)-1}dz 
\\=& 
\lim_{k\rightarrow \infty} \frac{1}{2\pi \ii}\int_{z_{cr_1}}^{z_{cr_2}} (-e^{-\gt})^{P_{\gl_k}(t_1^k,t_2^k)} (1-e^{-\gt} z)^{\gD t} z^{-\gD h + P_{\gl_k}(t_1^k,t_2^k)+\frac 12 b_{\lambda_k}(t_1^k)-\frac 12 b_{\lambda_k}(t_2^k)-1}dz
\\=& \lim_{k\rightarrow \infty} 
\frac{1}{2\pi \ii}\int_{z_{cr_1}}^{z_{cr_2}} (-e^{-\gt})^{P_{\gl_k}(t_1^k,t_2^k)} (1-e^{-\gt} z)^{\gD t} z^{-\gD h -\frac 12 \gD t-1}dz.
\end{align*}

In the case $t_2^k<t_1^k$, $K_{\lambda_k,q_k}$ has exactly the same expression as above, with $-P_{\gl_k}(t_1^k,t_2^k)$ instead of $P_{\gl_k}(t_1^k,t_2^k)$. It is easy to verify from the definition of $P_{\gl_k}(t_1^k,t_2^k)$ that it has the form $f(t_1^k)-f(t_2^k)$. Recall that the correlation functions are given by a determinant. Therefore, the process described by the kernel $K_{\lambda_k,q_k}$ is the same as the process with a kernel without the term $(-e^\gt)^{\pm P_{\gl_k}(t_1^k,t_2^k)}$. Thus, the following theorem is obtained:

\begin{theorem}\label{thm:cor_beta} 
The correlation functions of the system near a point $(\gt,\gq)$ in the bulk are given by the incomplete beta kernel 
\begin{equation}
\label{eq:CorrKer}
  K_{\gt,\gq}(\Delta t, \Delta h) = \int_\gamma (1- e^{-\tau} z)^{\Delta t} z^{-\Delta h -\frac{\Delta t}{2}-1} \frac{dz}{2\pi\ii z},
\end{equation}
where the integration contour connects the two non-real complex critical points of $S_{\tau,\chi}(z)$, passing through the real line in the interval $(0,1)$ if $\Delta(t)\geq 0$ and through $(-\infty, 0)$ otherwise.
\end{theorem} 

When $\gt,\gq$ are such that $S_{\gt,\gq}$ has no non-real complex critical points, there are two possibilities. Either along the original contours $C_z,C_w$ it is true that $\Re(S(z))\leq\Re(S(w))$, in which case the above arguments imply the correlation kernel is zero, or otherwise, when the contours are deformed to obtain $\Re(S(z))\leq\Re(S(w))$, one of the contours completely passes over the other and the correlation kernel is given again by \eqref{eq:CorrKer} with $\gamma$ now being a full circle around $0$ and such that $e^\tau$ is not inside it. Taking $\gD h = \gD t = 0$ gives that the limit of one point correlation functions, $\lim_{k\rightarrow\infty}\rho_{\lambda_k,q_k}(t_1^k,h_1^k)$ is either $0$ or $1$. This means that if $S_{\gt,\gq}(z)$ has no non-real complex critical points, the tiles near the point $\gt,\gq$ are horizontal tiles with probability $0$ or $1$, and therefore the point$(\gt,\gq)$ is in a frozen region. 

The above results imply that the curve $(\gt(z),\gq(z))$ where $S_{\gt,\gq}'(z)=S_{\gt,\gq}''(z)=0$ is the boundary between the frozen and liquid regions, so called frozen boundary. Section \ref{sec:TheBoundary} studies this curve in detail.

Theorem \ref{thm:IndOfFam} follows immediately from Theorem \ref{thm:cor_beta}, since \eqref{eq:CorrKer} gives that the correlation kernel depends only on the critical points of the function $S_{\tau,\chi}(z)$ which, as can be seen from its formula given in \eqref{eq:S_integral}, is independent of the sequence $\lambda_k$ of partitions and only depends on the function $V(\tau)$ giving the back wall in the scaling limit.

\subsection{Correlation kernels on the frozen boundary and when \texorpdfstring{$\gq\rightarrow\infty$}{chi is large}}
\label{subsec:CorrKerOther}
The correlation kernel on the frozen boundary is given by the Airy kernel and near the cusps by the Pearcey process. The calculations are almost identical to those in \cite{OR2}, therefore they will be omitted.

It is proven in Section 4.2 of \cite{BMRT} that if in the limit $\chi\rightarrow\infty$, $S(z)$ has one pair of complex conjugate critical points and they have the asymptotics given in Lemma \ref{lem:critPtsLargeChi}, then the local statistics of the system is described by the bead process of Boutillier \cite{Bou}.

The analysis in Section \ref{subsubsec:nearCorners} shows that in the limit $$\chi\rightarrow\infty,\gt=V_{j-1}+\gd,\gd\rightarrow 0,$$ if $$p=e^{\chi-\chi^{j-1}}|\gd|^{1-\frac 12(\gb_j-\gb_{j-1})}$$ is fixed, then depending on the type of the angle at $V_{j-1}$ and whether $p\ll1$ or $p\gg1$ either the bead process or a frozen region is observed. This is illustrated in Figures \ref{Fig:M-Bdry} and \ref{Fig:SM1-BoundaryComponents}.

\section{The boundary of the limit shape}
\label{sec:TheBoundary}
The purpose of this section is to study the frozen boundary. It consists of those points $(\gt,\gq)$ for which $S_{\gt,\gq}(z)$ has double real critical points. In this section the most useful representation of the frozen boundary is as the parametric curve $(\gt(z),\gq(z))$ where $\gt(z)$ and $\gq(z)$ are given by \eqref{eq:TauZ} and \eqref{eq:ChiZ}, and $z\in\RRR$ runs over all real numbers for which $\gt(z)$ and $\gq(z)$ are real.

Define $\UU$ to be 
$$\UU= \left((-\infty,e^{V_0})\union(e^{V_n},\infty)\right)\bigcup\left(\bigcup_{\stackrel{1\leq i\leq n}{\gb_i=\pm 1}} (e^{V_{i-1}},e^{V_i})\right).$$
It was established in Section \ref{subsec:ValidZ-s} that if $\gt(z)$ and $\gq(z)$ are real, then $z$ must be in $\UU$.
\subsection{\texorpdfstring{$\gt(z)\text{ and }\gq(z)$}{tau(z) and chi(z)} are real for all \texorpdfstring{$z\in\UU$}{z in U}}
\label{subsec:TauChiDefinedForAllValidZ}

\begin{proposition}
\label{prop:TauChiRealForAllZs}
For every $z\in\UU$, $\gt=\gt(z)$ and $\gq=\gq(z)$ given by \eqref{eq:TauZ} and \eqref{eq:ChiZ} are real. In particular, $\forall z\in\UU$ the point $(\gt(z),\gq(z))$ is on the frozen boundary.
\end{proposition}

Before proving this proposition we need to prove one lemma. 

Define a function $T(z)$ by 
\begin{equation}
\label{eq:T}
T(z):=\sum_{i=0}^n\frac 12(\gb_{i+1}-\gb_i)\frac 1{z-e^{V_i}}.
\end{equation}
$\gt(z)$ from \eqref{eq:TauZ} can be expressed in terms of $T(z)$ by
\begin{equation}
\label{eq:tauT}
e^{\gt(z)}=z-\frac 1{T(z)}.
\end{equation}

\begin{lemma}
\label{lem:T(z)neq0}
For any integer $m\geq 0$ the $2m$'th derivative of $T$, $T^{(2m)}(z)\neq 0$ for all $z\in\UU$.
\end{lemma}

\begin{proof}
Formula \eqref{eq:T} gives that 
$$T^{(2m)}(z)=\sum_{i=0}^n\frac 12(\gb_{i+1}-\gb_i) \frac{(-1)^{2m}(2m)!}{(z-e^{V_i})^{2m+1}}.$$
Suppose by contradiction that $\exists z_0\in\UU$ such that $T^{(2m)}(z_0)=0$. From the definition of $\UU$ it follows that $z_0\in(-\infty,e^{V_0})\union(e^{V_n},\infty)$ or $z_0\in(e^{V_{l-1}},e^{V_l})$ with $\gb_l=\pm 1$. The case $z_0\in(e^{V_{l-1}},e^{V_l})$, $\gb_l=1$ will be treated. Other cases are similar.

When there are no corners, i.e. when $n=1$, the proof is trivial. The rest of the proof is by induction on the number of corners of the limiting back wall. Suppose that there is only one corner, so $n=2$, $V_0<V_1<V_2$, $l=1$, $e^{V_0}<z<e^{V_1}$, $\gb_1=1$, and $\gb_2\in[-1,1)$. Then 
\begin{multline*}
T^{(2m)}(z_0)=\((-1)^{2m}(2m)!\)\frac 12\(\frac 2{z-e^{V_0}} + \frac{\gb_2-1}{z-e^{V_1}} + \frac{1-\gb_2}{z-e^{V_2}}\) \\= (2m)!\frac 12\(\frac 2{z-e^{V_0}} + \frac{(1-\gb_2)(e^{V_2}-e^{V_1})}{(z-e^{V_1})(z-e^{V_2})}\) >0,
\end{multline*}
which contradicts the assumption $T^{(2m)}(z_0)=0$. Thus, $T^{(2m)}(z_0)\neq 0$ when the back wall has a single corner.

Now, suppose the result is true when the number of corners is less than $n-1$, and prove it when the number of corners is $n-1$. First, consider the case $n>l+1$, i.e. when there is at least one corner to the right of the corner at position $V_l$. It will be shown that if $T^{(2m)}(z_0)=0$, the number of corners to the right of $z_0$ can be reduced and still have that $T^{(2m)}(z_0)=0$. 

\begin{figure}[ht]
\caption{\label{Fig:MoveVs} Reducing the number of corners.}
\includegraphics[width=10cm]{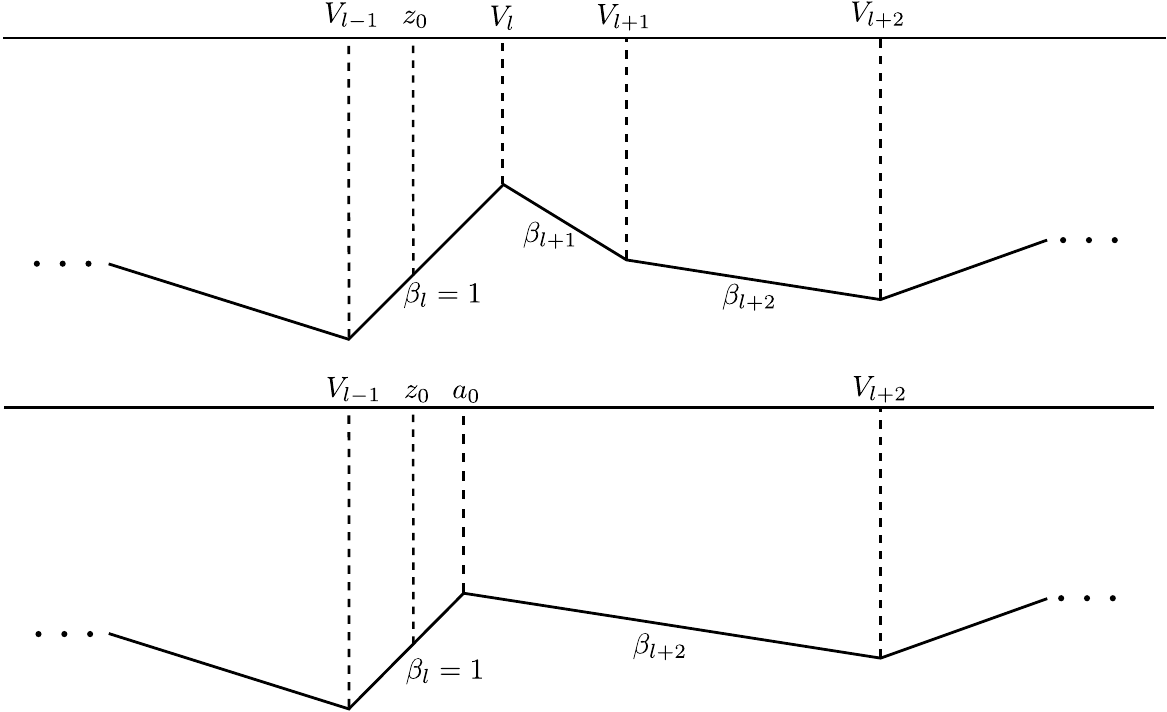}
\end{figure}

Consider the function $\tilde{T}(z)$ defined as
\begin{equation*}
\tilde{T}(z,a,b):=\sum_{\stackrel{i=0}{i\neq l,i\neq l+1}}^{n}\frac 12(\gb_{i+1}-\gb_i)\frac 1{z-e^{V_i}}+\frac 12(\gb_{l+1}-\gb_l)\frac 1{z-e^{a}}+\frac 12(\gb_{l+2}-\gb_{l+1})\frac 1{z-e^{b}}.
\end{equation*}
Notice that $\tilde{T}(z,V_{l},V_{l+1})=T(z)$. In particular, $\frac{\partial^{2m}}{\partial z^{2m}}\tilde{T}(z_0,V_{l},V_{l+1})=T^{(2m)}(z_0)=0$. Consider two cases.
\begin{enumerate}
\item[Case 1:] $\gb_{l+2}-\gb_{l+1}>0$. In this case $\frac{\partial^{2m}}{\partial z^{2m}}\tilde{T}(z_0,a,b)$ is an increasing function of $b$ for $b$ such that $e^b>z_0$ and thus $\frac{\partial^{2m}}{\partial z^{2m}}\tilde{T}(z_0,V_l,V_l)<0$. From $\beta_l=1$ it follows that $\frac 12(\beta_{l+2}-\beta_l)<0$ which implies that $\frac{\partial^{2m}}{\partial z^{2m}}\tilde{T}(z_0,a,a)$ is decreasing as a function of $a$ when $e^a>z_0$ and that $$\lim_{a\rightarrow\ln(z_0)+}\frac{\partial^{2m}}{\partial z^{2m}}\tilde{T}(z_0,a,a)=\infty.$$ 
Thus, $\exists a_0\in\RRR$ such that $z_0<e^{a_0}<e^{V_l}$ and $\frac{\partial^{2m}}{\partial z^{2m}}\tilde{T}(z_0,a_0,a_0)=0$. This contradicts the inductive assumption since $\tilde{T}(z,a_0,a_0)$ is equal to the function $T(z)$ corresponding to a piecewise linear back wall with parameters $V_0,V_1,\ldots,V_{l-1}$, $a_0$, $V_{l+2},\ldots,V_n$ and $\beta_1,\beta_2,\ldots,\beta_l,\beta_{l+2},\ldots,\beta_n$, which has at most $n-2$ corners (see Figure \ref{Fig:MoveVs}).
\item[Case 2:] $\gb_{l+2}-\gb_{l+1}<0$. The argument is similar to the previous case. Now, $\frac{\partial^{2m}}{\partial z^{2m}}\tilde{T}(z_0,a,b)$ is a decreasing function of $b$ for $b$ such that $e^b>z_0$ and thus $\frac{\partial^{2m}}{\partial z^{2m}}\tilde{T}(z_0,V_l,V_{l+2})<0$. $\frac{\partial^{2m}}{\partial z^{2m}}\tilde{T}(z_0,a,V_{l+2})$ is decreasing as a function of $a$ when $e^a>z_0$ and  $$\lim_{a\rightarrow\ln(z_0)+}\frac{\partial^{2m}}{\partial z^{2m}}\tilde{T}(z_0,a,V_{l+2})=\infty.$$ 
Thus, $\exists a_0\in\RRR$ such that $z_0<e^{a_0}<e^{V_l}$ and $\frac{\partial^{2m}}{\partial z^{2m}}\tilde{T}(z_0,a_0,V_{l+2})=0$. This contradicts the inductive assumption since $\tilde{T}(z,a_0,V_{l+2})$ is equal to the function $T(z)$ corresponding to a piecewise linear back wall with parameters $V_0,V_1,\ldots,V_{l-1}$, $a_0$, $V_{l+2},\ldots,V_n$ and $\beta_1,\beta_2,\ldots,\beta_{l+1},\beta_{l+3},\ldots,\beta_n$, which has at most $n-2$ corners.
\end{enumerate}

This concludes the proof when there is at least one corner to the right of the corner at position $V_l$. When there are corners to the left, the argument can be easily modified to show that if $T^{(2m)}(z_0)=0$, the number of corners to the left of $z_0$ can be reduced and still have that $T^{(2m)}(z_0)=0$. 
$\qed$
\end{proof}

\begin{proof}[Proof of Proposition \ref{prop:TauChiRealForAllZs}]
Again, only the case $z\in(e^{V_{l-1}},e^{V_l})$ with $\gb_l=1$ will be presented, as other cases work in the same way. From \eqref{eq:tauT}, to show $\gt$ is real it is enough to show that $z-\frac 1{T(z)}>0$. Lemma \ref{lem:T(z)neq0} gives that $T(z)\neq 0$. If $T(z)<0$, there is nothing to show, so assume $T(z)>0$. 

Again, do induction on the number of corners on the back wall. When there is only one corner, i.e. when $n=2$, $V_0<V_1<V_2$, $l=1$, $e^{V_0}<z<e^{V_1}$, $\gb_1=1$, and $\gb_2\in[-1,1)$, the inequality can be rewritten as follows:
\begin{multline*}
z-\frac 1{T(z)}>0 \Leftrightarrow 1<zT(z) \Leftrightarrow 1< \frac 12 z \(\frac 2{z-e^{V_0}} + \frac{\gb_2-1}{z-e^{V_1}} + \frac{1-\gb_2}{z-e^{V_2}}\) \\\Leftrightarrow \frac 12 z(1-\gb_2)\frac{e^{V_1}-e^{V_2}}{(z-e^{V_1})(z-e^{V_1})} < \frac z{z-e^{V_0}} - 1.
\end{multline*}
The last inequality is true since the LHS is less than zero and the RHS is greater than zero.

Now, assume that $z-\frac 1{T(z)}>0$ whenever the back wall has at most $n-2$ corners. Fix $z$ and let $0<l<n-1$ be such that $z<e^{V_l}$. Such $l$ exists if not all corners are to the left of $z$. The case when all corners are to the left of $z$ can be treated similarly. 

Consider the function $\tilde{T}(z,V_l,b)$ from the proof of Lemma \ref{lem:T(z)neq0}. Depending on the sign of $\beta_{l+2}-\beta_{l+1}$, $\tilde{T}(z,V_l,b)$ is either increasing or decreasing in $b$ for $b\geq V_l$. Hence, either $0<\tilde{T}(z,V_l,V_l)<T(z)$ or $0<\tilde{T}(z,V_l,V_{l+2})<T(z)$. However, as in Lemma \ref{lem:T(z)neq0}, both $\tilde{T}(z,V_l,V_l)$ and $\tilde{T}(z,V_l,V_{l+2})$ are equal to the function $T(z)$ corresponding to back walls with at most $n-2$ corners. Thus, the induction hypothesis gives that $z-\frac 1{\tilde{T}(z,V_l,V_{l+2})}>0$ and $z-\frac 1{\tilde{T}(z,V_l,V_l)}>0$. Hence, $z-\frac 1{T(z)}>0$. 

This establishes that if $z\in\UU$, then $\gt(z)=\ln(z-\frac 1{T(z)})$ is real.

Now, let us show that $\gq(z)$ is real. This is obvious if $z\in(-\infty,e^{V_0})\union(e^{V_n},\infty)$. Suppose $z\in(e^{V_{l-1}},e^{V_l})$ with $\gb_l=\pm 1$. If $\gb_l=1$, ($\gb_l=-1$), then $\gb_{l+1}-\gb_l<0$, ($\gb_{l+1}-\gb_l>0$). If $z=e^{V_l}-\varepsilon$, it follows from \eqref{eq:T} that $T(z)>0$, ($T(z)<0$). By Lemma \ref{lem:T(z)neq0}, $T(z)\neq 0$, which implies $T(z)>0$, ($T(z)<0$) for all $z\in(e^{V_{l-1}},e^{V_l})$. Now, \eqref{eq:tauT} implies that $e^\gt<z<e^{V_l}$, ($e^\gt>z>e^{V_{l-1}})$. Therefore, the coefficient of $\ln(\frac{ze^{V_l}-1}{ze^{V_{l-1}}-1})$ in \eqref{eq:ChiZ} is $1-\gb_l=0$, ($1+\gb_l=0$). In conclusion, it was shown that if $z\in(e^{V_{l-1}},e^{V_l})$ with $\gb_l=\pm 1$, then the coefficient of $\ln(\frac{ze^{V_l}-1}{ze^{V_{l-1}}-1})$ in \eqref{eq:ChiZ} is always zero, which implies that $\gq(z)$ is real for all $z\in\UU$.
$\qed$
\end{proof}

\subsection{The frozen boundary: connected components, cusps}
\label{subsec:TheFrozenBoundary}

\begin{proposition}
The frozen boundary consists of connected components, one for each outer corner where at least one of the slopes is $\pm 1$,  and one connected component at the bottom (see Figure \ref{Fig:M-Bdry}).
\end{proposition}

\begin{proof}
If $(\gt,\gq)$ is on the frozen boundary, then it is easy to see from \eqref{eq:ChiZ} that 
\begin{equation}
\label{eq:dChidTau}
\frac{d\gq}{d\gt}=\frac z{z-e^\gt} -\frac 12.
\end{equation}
Suppoze $z\in(e^{V_{l-1}},e^{V_l})\subset \UU$, for some $l$, for which $\gb_l=\pm 1$. Assume $\gb_l=1$ ($\gb_l=-1$ can be done similarly). It is immediate from the formulas for $\gt(z),\gq(z)$ that $(\gt(z),\gq(z))$ is a continuous curve when $z\in(e^{V_{l-1}},e^{V_l})$.
Consider two cases. 

First, suppose that $\gb_{l-1}\neq -1$. Then \eqref{eq:TauZ} implies
\begin{equation}
\label{eq:limTauZtoV}
\lim_{z\rightarrow e^{V_{l}}-}\gt(z)=V_{l}-, \text{ and } 
\lim_{z\rightarrow e^{V_{l-1}}+}\gt(z)= V_{l-1}-,
\end{equation}
where for example by $\lim_{z\rightarrow e^{V_{l}}-}\gt(z)=V_{l}-$ we mean $\gt(z)$ converges to $V_{l}$ from below when $z$ converges to $e^{V_{l}}$ from below.

Looking at \eqref{eq:ChiZ} it is easy to see that $$\lim_{z\rightarrow {e^{V_l}\text{ or }e^{V_{l-1}}}}\gq(z) =+\infty.$$ 
These calculations imply that when $z$ ranges over $(e^{V_{l-1}},e^{V_l})$, it gives a connected component of the curve $(\gt(z),\gq(z))$.

If $\gb_{l-1}=-1$, then \eqref{eq:TauZ} gives that $$\lim_{z\rightarrow e^{V_{l-1}}-}\gt(z)=\lim_{z\rightarrow e^{V_{l-1}}+}\gt(z)=V_{l-1}.$$
Consider values of $z$ which range over $(e^{V_{l-2}},e^{V_{l-1}})\cup(e^{V_{l-1}},e^{V_l})$. It follows from \eqref{eq:ChiZ} that
\begin{align*}
\gq_0:=\lim_{z\rightarrow V_{l-1}\pm}\gq(z)
=&-\sum_{i=1}^{l-2}\frac 12(1+\gb_i) \ln\(\frac{ze^{-V_i}-1}{ze^{-V_{i-1}}-1}\)
-\\
&+\sum_{i=l+1}^{n}\frac 12(1-\gb_i) \ln\(\frac{ze^{-V_i}-1}{ze^{-V_{i-1}}-1}\) -\frac 12 V_{l-1}+\frac 12(V(V_0)+V_0)
\end{align*}
and
$$\lim_{z\rightarrow {e^{V_l}\text{ or }e^{V_{l-2}}}}\gq(z) =+\infty.$$

These imply that when $z$ ranges over $(e^{V_{l-2}},e^{V_{l-1}})\cup(e^{V_{l-1}},e^{V_l})$, one connected component of the curve $(\gt(z),\gq(z))$ is obtained. The two pieces of this connected component corresponding to the intervals $z\in(e^{V_{l-2}},e^{V_{l-1}})$ and $z\in(e^{V_{l-1}},e^{V_l})$ are connected at the point $(V_{l-1},\gq_0)$.

Values $z\notin[e^{V_{0}},e^{V_{n}}]$ correspond to the component at the bottom.
$\qed$
\end{proof}

\begin{figure}[ht]
\caption{\label{Fig:SM1-BoundaryComponents}All possible corners.}
\includegraphics[width=15cm]{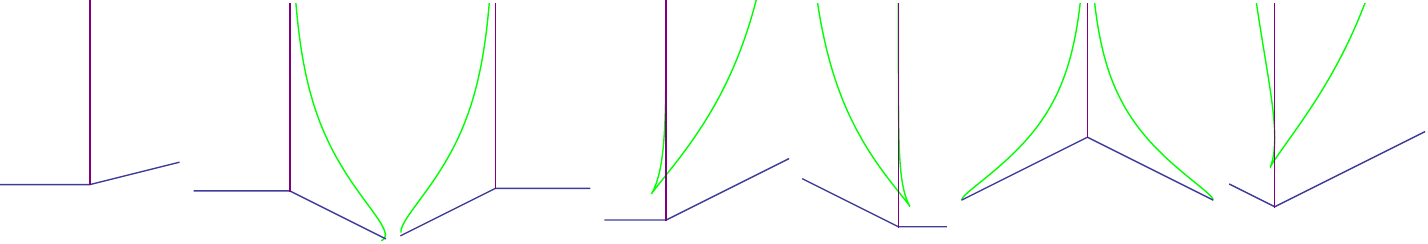}
\end{figure}

Next, it will be shown that each component corresponding to an outer corner develops a cusp. The above results, together with the analysis in Section \ref{subsec:CorrKerOther} imply the frozen boundary at corners looks as described in Figure \ref{Fig:SM1-BoundaryComponents}.

A cusp appears if $\frac{d\gq}{dz}=\frac{d\gt}{d z}=0$. This is equivalent to $S'(z)=S''(z)=S'''(z)=0$.

\begin{lemma}
\label{lem:cusp-cond}
Suppose $(\gt,\gq)$ is on the frozen boundary. Then $$S'''(z)=0 \Leftrightarrow \frac{d\gt}{d z}=0.$$
\end{lemma}

\begin{proof}
By definition of $T(z)$ $$e^\gt=z-\frac 1{T(z)},$$ so
$$\frac{d\gt}{d z}=0 \Leftrightarrow 1+\frac {T'}{T^2}=0 \Leftrightarrow T'+T^2=0 \Leftrightarrow 0=T'+\frac 1{(z-e^\gt)^2}. $$
However, notice that $S'''(z)=T'+\frac 1{(z-e^\gt)^2}$. Thus, $\frac{d\tau}{d z}=0\Leftrightarrow S'''(z)=0$.
$\qed$
\end{proof}

\begin{theorem}
Each connected component of the frozen boundary corresponding to an outer corner has a cusp on it (see Figures \ref{Fig:M-Bdry} and \ref{Fig:SM1-BoundaryComponents}). In particular, the number of cusps is equal to the number of outer corners where at least one of the slopes is a lattice slope. 
\end{theorem}

\begin{proof}
First, let us treat the case $e^{V_{l-1}}<z<e^{V_l}$, $\gb_{l-1}>-1, \gb_{l}=1$.

By Lemma \ref{lem:cusp-cond} cusps correspond to points $z$ such that $$e^\gt=z-\frac 1{T(z)}\text{ and } \frac{d\gt}{dz}=0.$$
This is equivalent to $$T(z)=\frac{1}{z-e^\gt}\text{ and }T'(z)=\(\frac{1}{z-e^\gt}\)',$$
which means that to prove the theorem it is enough to show that when $e^{V_{l-1}}<z<e^{V_l}$, then $T(z)$ and $\frac{1}{z-e^\gt}$ are tangent to each other at a single point.

Recall that $T(z)$ is given by $$T(z)=\sum_{i=0}^n\frac 12(\gb_{i+1}-\gb_i)\frac 1{z-e^{V_i}}.$$
It follows that $\lim_{z\rightarrow e^{V_{l-1}}+}T(z)=\infty=\lim_{z\rightarrow e^{V_{l}}-}T(z)$. Lemma \ref{lem:T(z)neq0} gives that $T^{(2m)}(z)\neq 0$. In particular it follows that in the interval $e^{V_{l-1}}<z<e^{V_l}$ the function $T(z)$ is never zero, is always concave up, and has a positive minimum in that interval.

\begin{figure}[ht]
\caption{\label{Fig:TforCusps}}
\includegraphics[width=7cm]{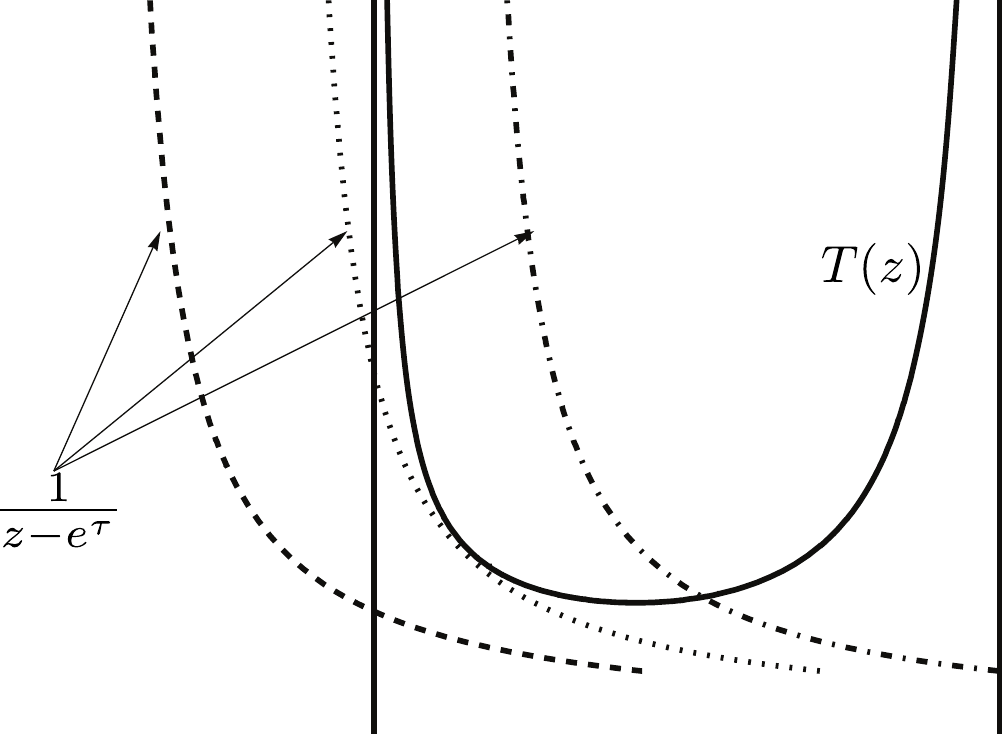}
\end{figure}

When $\gt>V_l$, $\frac{1}{z-e^\gt}$ does not intersect $T(z)$ in the interval $e^{V_{l-1}}<z<e^{V_l}$. When $\gt$ decreases, it intersects at a point. When $\gt=V_{l-1}$, then
$$\lim_{z\rightarrow e^{V_{l-1}}+}\frac{1}{z-e^{V_{l-1}}}-T(z)= \lim_{z\rightarrow e^{V_{l-1}}+} \frac{1}{z-e^{V_{l-1}}} - \frac{\frac 12 (1-\gb_{l-1})}{z-e^{V_{l-1}}} +\text{ finite terms }=\infty,$$
so $\frac{1}{z-e^\gt}$ and $T(z)$ intersect. $T(z)$ has a positive minimum in the interval $e^{V_{l-1}}<z<e^{V_l}$, so if $\gt$ is small enough, $\frac{1}{z-e^\gt}$ and $T(z)$ will be tangent to each other. This completes the proof in this case.

The case  $e^{V_{l-1}}<z<e^{V_l}$, $\gb_{l-1}=-1, \gb_{l}<1$ is identical.

It remains to study the case $e^{V_{l-1}}<z<e^{V_l}$, $\gb_{l-1}=-1, \gb_{l}=1$. 
The same thing happens as above, except now the infinite terms in 
$$\lim_{z\rightarrow e^{V_{l-1}}+}\frac{1}{z-e^{V_{l-1}}}-T(z)$$
cancel each other. If the limit is still positive, the argument works the same way as above, and a cusp $\gt<V_{l-1}$ is obtained. If the limit is negative, then as above it can be shown that a cusp $\gt>V_{l-1}$ appears when $e^{V_{l-2}}<z<e^{V_{l-1}}$. If the limit is zero, the cusp is at $\gt=V_{l-1}$.
$\qed$
\end{proof}

Notice that \eqref{eq:limTauZtoV} implies $\frac{d\gq}{d\gt}\rightarrow\pm\infty$ when $z\rightarrow e^{V_l},\forall l$. 
Examining \eqref{eq:dChidTau} it can also be established that $\frac{d\gq}{d\gt}\neq 0$, and since for $z\in(e^{V_{l-1}},e^{V_l})$, $\frac{d\gq(\gt)}{d\gt}$ is a continuous function, it follows that in the intervals $z\in(e^{V_{l-1}},e^{V_l})$, $\frac{d\gq(\gt)}{d\gt}$ does not change its sign.

In the $\gt,\gq$ plane each connected component except the one at the bottom starts at $\gt=V_{i-1},\gq=\infty$ for some $i$, $\gq$ decreases until it reaches the cusp, then increases to $\gt=V_i,\gq=\infty$. See Figures \ref{Fig:M-Bdry} and \ref{Fig:SM1-BoundaryComponents}.

\begin{figure}[ht]
\caption{\label{Fig:2cusps}
$V = {0, 1, 1.05, 2, 2.05, 3, 3.05},\gb = {1, -1, 1, 0.7, 1, 0.7}$.}
\includegraphics[width=7cm]{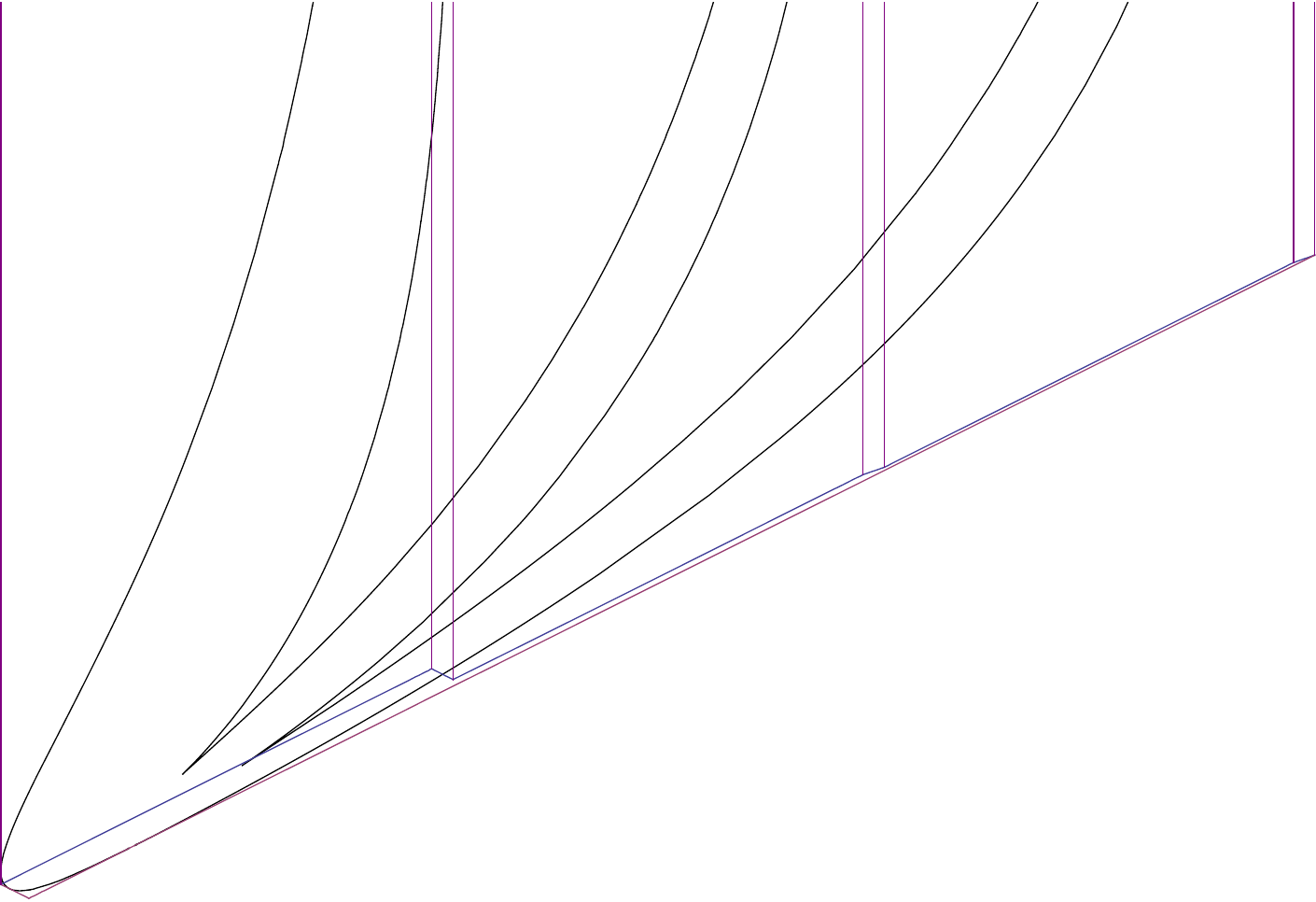}
\end{figure}

The results of this and previous sections give the following characterization of the frozen boundary:
\begin{enumerate}
\item On the pieces of the back wall where the slope is $\pm 1$ the disordered region is bounded above, while at other places it grows infinitely high. 
\item The number of connected components of the frozen boundary is one more than the number of outer corners where at least one of the slopes at the corner is a lattice slope.
\item The frozen boundary develops a cusp for each such outer corner. 
\end{enumerate}

Figures \ref{Fig:CuspsPass}, \ref{Fig:2cusps} and \ref{Fig:3cusps} give several interesting examples of frozen boundaries.

\begin{figure}[ht]
\caption{\label{Fig:3cusps}
$V = {0.7, 1, 1.05, 1.2, 1.25, 1.5, 1.55, 3, 3.1}, \gb = {   1, 0.7,  1,  0.7,     1, 0.7, 1, -1}$.}
\includegraphics[width=6cm]{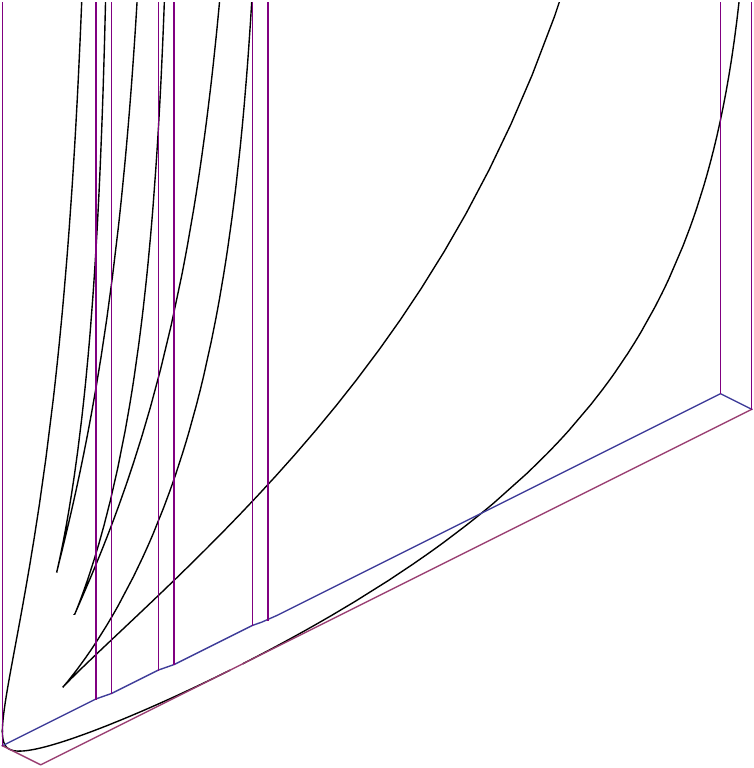}
\end{figure}

\appendix

\section{An integral formula for \texorpdfstring{$\Phi_{b_k}(z)$}{Phi\_b\_k(z)}} \label{sec:S_appendix}

One of the main difficulties in this paper, and in fact also in \cite{OR1} and \cite{OR2}, is understanding the asymptotics of \eqref{eq:main-corr2}. Similarly as in \cite{OR1} we study the asymptotics of the function
\begin{equation*}
r_k\ln\Phi_{b_{\lambda_k}}(z,t_k)=r_k\ln\frac{\Phi_{-,b_{\lambda_k}}(z,t_k)}{\Phi_{+,b_{\lambda_k}}(z,t_k)}
\end{equation*}
in the limit when $r_k\rightarrow 0$ and $t_k$ is a sequence such that $r_kt_k\rightarrow \tau$ (here $\Phi_{-,b_{\lambda_k}}(z,t)$ and $\Phi_{+,b_{\lambda_k}}(z,t)$ are as in \eqref{eq:Phis}). In this appendix we prove a technical result concerning this asymptotics in a very general setting. We consider an arbitrary continuous function 
\begin{equation*}
V(\tau): {\mathbb R} \rightarrow {\mathbb R},
\end{equation*}
which is Lipshitz with constant 1. This is a less restrictive assumption on $V$ than in the rest of the paper, where $V(\tau)$ is assumed to be piecewise linear, and the domain is restricted to a finite interval. Otherwise, we use notation from the introduction. In particular, we are considering a sequence of back walls $b_{\lambda_k}(t)$ and $r_k \in {\mathbb R}_{>0}$ such that $r_k$ tends to $0$, and functions $B_{\lambda_k}(\tau)$ defined by 
\begin{equation*}
B_{\lambda_k}(\tau) := r_k b_{\lambda_k}\(\frac{\tau}{r_k}\)
\end{equation*}
converge uniformly to $V(\tau)$. 

\begin{lemma} \label{lem-IndOfFam} If $B_{\lambda_k}$ uniformly converge to $V$, $\lim_{k\rightarrow\infty}r_k=0$, and $\lim_{k\rightarrow\infty}r_kt_k=\tau$, then 
\begin{multline*}
\lim_{k\rightarrow \infty}r_k \ln \Phi_{b_{\lambda_k}}(z,t_k)
=\frac12(\tau+V(\tau))\ln\(1-e^{\tau}z^{-1}\)
-\int_{-\infty}^{\tau}\frac{1}{2}(M+V(M))\frac{-e^Mz^{-1}}{1-e^{M}z^{-1}}dM
\\+\frac12(\tau-V(\tau))\ln\(1-e^{-\tau}z\)+\int_\tau^\infty\frac{1}{2}(M-V(M))\frac{e^{-M}z}{1-e^{-M}z}dM.
\end{multline*}
In particular $\lim_{k\rightarrow \infty} r_k \ln \Phi_{b_{\lambda_k}}(z,t_k)$ is independent of the family $\{b_{\lambda_k}(t)\}_k$. The convergence is uniform in $z$ on compact subsets of $\mathbb{C}$.
\end{lemma}

\begin{proof} 
Let us analyze the denominator of $\Phi_{b_{\lambda_k}}(z,t_k)$. The same arguments will work for the numerator. Define 
$$D_k:=r_k\ln\(\frac{1}{\Phi_{+,b_{\lambda_k}}(z,t_k)}\).$$
From \eqref{eq:Phis} we get
$$D_k=-r_k\(\ln\(\prod_{m>t_k, m\in D^+, m\in \ZZZ+\frac 12}(1-zq_k^m)\)\)=r_k\(-\sum_{m>t_k,b'_k(m)=-1}\ln\(1-e^{-r_km}z\)\).$$
Notice that 
\begin{equation}
\label{eq:b'toMeas}
\frac{1}{2}(1-b'_k(m))=\left\{
\begin{array}{ll}
1,&b'_k(m)=-1\\
0,&b'_k(m)=1
\end{array}\right..
\end{equation}
Make the change of variable $M=r_km$ and set $\tau_k=r_k(t_k+\frac 12)$. $D_k$ can be rewritten as
\begin{multline*}
D_k=-\sum_{m>t_k}r_k\frac{1}{2}(1-b'_k(m)) \ln\(1-e^{-r_km}z\)= \\-\sum_{M\in\{\tau_k,\tau_k+r_k,\tau_k+2r_k,\ldots\}} r_k\frac{1}{2}(1-B'_k(M)) \ln\(1-e^{-M}z\).$$
\end{multline*}

To make formulas simpler, define $$f(M):=-\ln\(1-e^{-M}z\).$$ 

Let $\sigma>\tau$, and let $\sigma_k$ be a sequence such that $\frac{\sigma_k-\tau_k}{r_k}\in\ZZZ_+$ and $\lim_{k\rightarrow\infty}\sigma_k=\sigma$. Define
$$D^\sigma_k:= \sum_{M\in\{\tau_k,\tau_k+r_k,\tau_k+2r_k,\ldots,\sigma_k\}} r_k\frac{1}{2}(1-B'_k(M))f(M).$$
Since $B_{\lambda_k}(s)$ has constant slope when $s\in(M,M+r_k)$ we can write $D^\sigma_k$ as
\begin{multline*}
D^\sigma_k= \sum_{M\in\{\tau_k,\tau_k+r_k,\tau_k+2r_k,\ldots,\sigma_k\}} r_k\frac{1}{2}\(1-\frac{B_{\lambda_k}(M+r_k)-B_{\lambda_k}(M)}{r_k}\)f(M)
=\\=\frac 12\sum_{M\in\{\tau_k,\tau_k+r_k,\tau_k+2r_k,\ldots,\sigma_k\}} ([M+r_k-B_{\lambda_k}(M+r_k)]-[M-B_{\lambda_k}(M)])f(M).
\end{multline*}

Now, using $(a_{n+1}-a_n)b_n=(a_{n+1}b_{n+1}-a_nb_n)-a_{n+1}(b_{n+1}-b_n)$ rewrite $D^\sigma_k$ as
\begin{multline*}
D^\sigma_k= \frac 12\sum_{M\in\{\tau_k,\tau_k+r_k,\tau_k+2r_k,\ldots,\sigma_k\}} [(M+r_k-B_{\lambda_k}(M+r_k))f(M+r_k)-(M-B_{\lambda_k}(M))f(M)]
\\-\frac 12\sum_{M\in\{\tau_k,\tau_k+r_k,\tau_k+2r_k,\ldots,\sigma_k\}} (M+r_k-B_{\lambda_k}(M+r_k))(f(M+r_k)-f(M)).
\end{multline*}
Summing up the telescoping sum and rewriting the second one give
\begin{multline*}
D^\sigma_k= \frac 12(\sigma_k+r_k-B_{\lambda_k}(\sigma_k+r_k))f(\sigma_k+r_k)-\frac 12(\tau_k-B_{\lambda_k}(\tau_k))f(\tau_k)\\-\frac 12\sum_{M\in\{\tau_k,\tau_k+r_k,\tau_k+2r_k,\ldots,\sigma_k\}} r_k(M+r_k-B_{\lambda_k}(M+r_k))\(\frac{f(M+r_k)-f(M)}{r_k}\).
\end{multline*}

Since $B_{\lambda_k}(M)$ converge uniformly to $V(M)$, and $f'(M)$ and $f''(M)$ are bounded when $M\in(\tau_k,\sigma_k)$, we have
\begin{multline*}
\lim_{k\rightarrow \infty}\left|D^\sigma_k-\( \frac 12(\sigma_k-V(\sigma_k))f(\sigma_k)-\frac 12(\tau_k-V(\tau_k))f(\tau_k)\right.\right.\\\left.\left.-\frac 12\sum_{M\in\{\tau_k,\tau_k+r_k,\tau_k+2r_k,\ldots,\sigma_k\}} r_k(M-V(M))f'(M)\)\right|=0.
\end{multline*}

The sum above is a Riemann sum for $\frac 12\int_\tau^\sigma(M-V(M))f'(M)dM$ and thus in the limit $r_k\rightarrow 0$ it converges to the integral, so we have
\begin{equation*}
\lim_{k\rightarrow \infty}D^\sigma_k=\( \frac 12(\sigma-V(\sigma))f(\sigma)-\frac 12(\tau-V(\tau))f(\tau)-\frac 12\int_\tau^\sigma (M-V(M))f'(M)dM\).
\end{equation*}

Since $\lim_{\sigma\rightarrow\infty} \frac 12(\sigma-V(\sigma))f(\sigma)=0$, we obtain 
$$\lim_{k\rightarrow \infty}D_k=\(-\frac12(\tau-V(\tau))f(\tau)-\int_\tau^\infty\frac{1}{2}(M-V(M))f'(M)dM\).$$ 

Similarly, if $N_k$ is defined as $$N_k:=r_k\(\ln\(\Phi_{-,b_{\lambda_k}}(z,t_k)\)\),$$ it can be shown that 
$$\lim_{k\rightarrow \infty}N_k=+\frac12(\tau+V(\tau))\ln\(1-e^{\tau}z^{-1}\)
-\int_{-\infty}^{\tau}\frac{1}{2}(M+V(M))\frac{-e^Mz^{-1}}{1-e^{M}z^{-1}}dM.$$

Combining the results completes the proof.
$\qed$
\end{proof}

\begin{corollary}
\label{cor:limPhi}
If $V(\tau)$ is piecewise differentiable, $\lim_{k\rightarrow\infty}r_k=0$, and $\lim_{k\rightarrow\infty}r_kt_k=\tau$, then
\begin{multline*}
\lim_{k\rightarrow \infty}r_k\ln\Phi_{b_{\lambda_k}}(z,t_k)=
-\int_{-\infty}^{\tau} -\frac{1}{2}(1+V'(M))\ln\(1-e^{M}z^{-1}\)dM\\+\int_\tau^\infty-\frac{1}{2}(1-V'(M)) \ln\(1-e^{-M}z\)dM.
\end{multline*}
\end{corollary}
\begin{proof}
This follows from Lemma \ref{lem-IndOfFam} by integration by parts.
$\qed$
\end{proof}

\bibliography{ArbitrarySlopes-2010-12-27}
\bibliographystyle{alpha}

\end{document}